\documentclass{article}
\usepackage[utf8]{inputenc}
\usepackage{comment}
\usepackage[dvipsnames]{xcolor}
\usepackage[colorlinks = true,
            linkcolor = blue,
            urlcolor  = blue,
            citecolor = blue,
            anchorcolor = blue
            hydelinks]{hyperref}
\newcommand{\changeurlcolor}[1]
{\hypersetup{urlcolor=#1}}    
\usepackage[english]{babel}
\usepackage{graphicx}
\usepackage{wrapfig}
\usepackage{caption}
\usepackage{subcaption}
\usepackage{rotating}
\usepackage{float}  
\usepackage{longtable}
\graphicspath{ {figures/} }
\usepackage{array}
\usepackage{lscape}
\usepackage{epstopdf}
\usepackage{amsthm}
\usepackage{amsmath}
\usepackage{enumerate}
\usepackage{amssymb}
\usepackage{tabularx}
\usepackage{xparse,mathtools}
\usepackage{authblk}
\usepackage{a4wide}
\makeatletter

\renewcommand*\env@matrix[1][*\c@MaxMatrixCols c]{%
  \hskip -\arraycolsep
  \let\@ifnextchar\new@ifnextchar
 \array{#1}}
\makeatother
\usepackage{systeme}

\usepackage[nottoc]{tocbibind}
\usepackage{csquotes}

\hypersetup{
	bookmarks=true,         % show bookmarks bar?
	unicode=false,          % non-Latin characters in Acrobat‚????s bookmarks
	pdftoolbar=true,        % show Acrobat‚????s toolbar?
	pdfmenubar=true,        % show Acrobat‚????s menu?
	pdffitwindow=false,     % window fit to page when opened
	pdfstartview={FitH},    % fits the width of the page to the window
	pdftitle={My title},    % title
	pdfauthor={Author},     % author
	pdfsubject={Subject},   % subject of the document
	pdfcreator={Creator},   % creator of the document
	pdfproducer={Producer}, % producer of the document
	pdfkeywords={keywords}, % list of keywords
	pdfnewwindow=true,      % links in new window
	colorlinks=true,       % false: boxed links; true: colored links
	linkcolor=blue,          % color of internal links
	citecolor=blue,        % color of links to bibliography
	filecolor=magenta,      % color of file links
	urlcolor=cyan           % color of external links
}

\newtheorem{theorem}{Theorem}
\newtheorem{lemma}{Lemma}
\newtheorem{corollary}{Corollary}
\newtheorem{proposition}{Proposition}

\newtheorem*{remark}{Remark}

\title{Cost optimisation of individual-based institutional reward incentives for promoting cooperation in finite populations}
\author[$\dagger$]{M. H. Duong}
\author[$\dagger$]{C. M. Durbac}
\author[$\ddag$]{T. A. Han}

\affil[$\dagger$]{School of Mathematics, University of Birmingham, UK.}
\affil[$\ddag$]{ School of Computing, Engineering and Digital Technologies, Teesside University, UK. }
\date\today

\begin{document} 
\maketitle
\begin{abstract}
In this paper, we study the problem of cost optimisation of individual-based institutional incentives (reward, punishment, and hybrid) for  guaranteeing  a certain minimal level of cooperative behaviour in a well-mixed, finite population. In this scheme, the individuals in the population interact via cooperation dilemmas (Donation Game or Public Goods Game) in which institutional reward is carried out only if cooperation is not abundant enough  (i.e., the number of cooperators is below a threshold $1\leq t\leq N-1$, where $N$ is the population size);  and similarly, institutional punishment is carried out only  when defection is too abundant. We study analytically the cases $t=1$ for the reward incentive under the small mutation limit assumption and two different initial states, showing that the cost function is always non-decreasing. We derive the neutral drift and strong selection limits when the intensity of selection tends to zero and infinity, respectively. We numerically investigate the problem for other values of $t$ and for population dynamics with arbitrary mutation rates.
%We also obtain the critical threshold of $t$ at which the monotonicity of the cost function changes in all cases. Our analytical results are illustrated with numerical investigations.
%Overall, our analysis provides novel theoretical insights into the design of cost-efficient institutional incentive mechanisms for promoting the evolution of cooperation in stochastic systems.
\end{abstract}
\newpage
\tableofcontents

\section{Introduction}
Cooperation refers to the act of paying a cost to oneself  in order to convey a benefit to somebody else.  It is one of the cornerstones of human civilisation and one of the reasons for our unprecedented success as a species \cite{atkins2019prosocial}. Organisations are constantly faced with the problem of allocating resources in a budget-effective way. This issue becomes particularly essential for institutions like local governments and  the United Nations, where the optimisation of resources is crucial in facilitating cooperative endeavors. Given the paramount importance of cooperation, these organisations are tasked with strategically managing the costs associated with incentivising collective efforts \cite{ostrom2009understanding,van2014reward}. 

A well-established theoretical framework for analysing the promotion of cooperation is Evolutionary Game Theory (EGT) \cite{sigmund2010calculus}, which has been used  in both deterministic and stochastic settings. Using this framework, several mechanisms for promoting the evolution of cooperation have been studied including kin selection, direct reciprocity, indirect reciprocity, network reciprocity, group selection, and different forms of  incentives \cite{nowak2006,sigmund2010calculus,perc2017statistical,rand2013human,van2014reward,xia2023reputation,hu2020rewarding,moralpref}. 

The current work focuses on institutional incentives \cite{sasaki2012take,sigmundinstitutions,wang2019exploring,duong2021cost,cimpeanu2021cost,sun2021combination,van2014reward,gurerk,gois2019reward,sun2021combination,liu2022effects,flores2024evolution,wang2024partial,hua2024coevolutionary}, which are a plan of action involving the use of reward (i.e., increasing the payoff of cooperators), punishment (i.e., decreasing the payoff of defectors), or a combination of the two, by an external decision-maker. More precisely, we study how the aforementioned institutional incentives can be used in a cost-efficient way for maximising the levels of cooperative behaviour in a population of self-regarding individuals. In the literature, although there is a significant amount of works using agent-based numerical simulations, there are only a few papers that employ a rigorous analysis of the problem at hand~\cite{wang2019exploring,han2018cost, duong2021cost,DuongDurbacHan2022, wangdecentralised,wang2023optimization}. The works that do use analysis employ two complementary approaches, either in a continuous setting where the evolutionary processes are modelled as a continuous dynamical system (for instance, using the replicator dynamics) \cite{wang2019exploring,wangdecentralised,wang2023optimization} or in a discrete setting, in which the population dynamics is modelled as a Markov chain \cite{han2018cost, duong2021cost,DuongDurbacHan2022}. We review in more detail both approaches in the next paragraphs since they are most relevant to the present work.

In the discrete setting, the evolutionary process is often described by a Markov chain with an update rule (for instance, in the absence of mutation, it is an imitation process with the Fermi strategy update rule). For well-mixed finite populations and general two-player two-strategy games, the problem of promoting the evolution of cooperative behaviour with a minimum cost is formulated and numerically studied in \cite{han2018cost}. 
In this paper, the decision-maker may use a full-invest scheme, in which, in each generation, all cooperators are rewarded (or all defectors are punished) or an individual-based scheme, in which cooperators (defectors) are rewarded (punished, respectively) only if cooperation is not frequent enough (defection is too abundant, respectively). In both cases, the expected total cost of interference is a finite sum, over the state space of the Markov chain, of per generation costs. In \cite{duong2021cost}, the authors then analyse the cost function for a full-invest scheme in which individuals interact via donation games or public goods games. They prove that the cost function exhibits a phase transition when the intensity of selection varies and exactly calculate the optimal cost of the incentive for any given intensity of selection. In a more recent paper \cite{DuongDurbacHan2022}, similar results are obtained in the case of hybrid (mixed) reward and punishment incentives.

%Another relevant study to this paper is Han and Tran-Thanh \cite{han2018cost}. The authors derived analytical conditions for which a general incentive scheme can guarantee a given level of cooperation while at the same time minimising the total cost of investment. As in the papers by Duong and Han \cite{duong2021cost} and Duong et al. \cite{DuongDurbacHan2022}, these results are highly sensitive to the intensity of selection. The authors also studied a class of incentive strategies that make an investment whenever the number of players with a desired behaviour reaches a certain threshold $t\in\{1,\ldots,N-1\}$ ($N$ is the population size), showing that there is a wide range of values for the threshold that outperforms standard institutional incentive strategies - those which invest in all players, i.e., the threshold is $t=N-1$ \cite{carrotstick}. While the aforementioned papers, by Duong and Han \cite{duong2021cost} and Duong et al. \cite{DuongDurbacHan2022} did not investigate a merit-based incentive scheme, this present one at hand does. Moreover, the paper by Han and Tran-Thanh \cite{han2018cost} did not \textit{analytically} study the cost-efficiency of the reward incentive scheme, while this work does. We analytically study the case for $t = 1$ to obtain the neutral drift and strong selection limits, as well as numerically show the existence of a phase transition. We numerically study the other cases, i.e., for $t > 1$.

In the continuous setting, the evolutionary process is modelled by the replicator dynamics, which is a set of differential equations describing the evolution of the behavioural frequencies. For infinitely large, well-mixed populations, the problem of providing a minimum cost that guarantees a sufficient level of cooperation is formulated as an optimal control problem in \cite{wang2019exploring}. By using the approach of the Hamilton-Jacobi-Bellman equation, the authors theoretically obtain the optimal reward (positive) or punishment (negative) incentive strategies with the minimal cumulative cost, respectively. Similar results for structured population are obtained in \cite{wang2023optimization} (for either reward or punish incentive separately) and in \cite{wang2023optimally} (for combined incentives), using pair approximate methods. In these papers, the decision-maker implements incentives centrally within the whole population.  In \cite{wangdecentralised}, the authors consider the same problem but using decentralised incentives, that is, in each game group, a local incentive-providing institution implements local punishment or reward incentives on group members. We also refer the reader to  \cite{wang2021} for a recent survey on this optimal control approach.

\paragraph{Overview of contribution of this paper.} Following the discrete approach in \cite{han2018cost,duong2021cost,DuongDurbacHan2022}, in this work, we rigorously study the problem of cost optimisation of institutional reward or/and punishment  for maximising the levels of cooperative behaviour (or guaranteeing at least a certain level of cooperation) for well-mixed, finite populations. We focus on individual-based schemes, in which, in each generation, only when the number of cooperators/defectors is below/above a certain threshold $t$ (with $1\leq t< N-1$, where $N$ is the population size), are they rewarded (punished, respectively) (the case  of the full-invest scheme $t=N-1$ has been studied in \cite{duong2021cost,DuongDurbacHan2022}).

Analysing this problem for an arbitrary value of $t$ would be very challenging due to the number of parameters involved such as the number of individuals in the population, the strength of selection, the game-specific quantities, as well as the efficiency ratios of providing the corresponding incentive. In particular, the Markov chain based evolutionary process is of order equal to the population size, which is large but finite. The calculation of the entries of the corresponding fundamental matrix, which appear in the cost function, is intricate, both analytically and computationally.

Our present work provides a rigorous analysis of this problem in the case of reward  for $t = 1$. The main analytical results of the paper can be summarised as follows.
\begin{enumerate}[(i)]
    \item We show that the cost function is always non-decreasing regardless of the values of other parameters, for two initial state assumptions, namely when the evolutionary dynamics starts either equally from the two homogeneous states or from the all-defector homogeneous state. The monotonicity of the reward cost function $E_r$ with respect to the incentive cost per capita $\theta$ highlights that, in order to achieve a higher level of cooperation, the institution needs to employ more financial resources.
    \item We obtain the asymptotic behaviour of the cost function in the limits of neutral drift and strong selection when the intensity of selection goes to zero or infinity, respectively.
    \end{enumerate}
We also numerically investigate the cost function and its behaviour for other values of $t$. While the main analytical results focus on the small mutation limit $\mu\rightarrow 0$, we perform numerical simulations for the case of arbitrary mutation rates.

The rest of the paper is organised as follows. In Section \ref{sec: models}, we present the model and methods. Our main results are Theorems \ref{thm: derivative positive t1 (1/2, 1/2)} and \ref{thm: derivative positive t1 (1, 0)} on the monotonicity of the reward cost function for $t = 1$ and Propositions \ref{prop: neutral drift t=1 (1/2, 1/2)}, \ref{prop: infinite selection t=1 (1/2, 1/2)}, \ref{prop: neutral drift t=1 (1, 0)}, and \ref{prop: infinite selection t=1 (1, 0)} on the asymptotic behaviour (neutral drift, strong selection limits) of the reward cost function for $t = 1$, in Section \ref{sec: reward}. Section \ref{sec: general mutation} contains the cost function (for reward, punishment, and hybrid incentives) in the case of a general mutation rate. In Section \ref{sec: numerical analysis}, we provide a numerical analysis of the reward cost function for both the small mutation limit and the general mutation rate cases, highlighting its qualitative behaviour. Summary and further discussions are provided in Section \ref{sec: discussion}. Finally, Section \ref{sec: appendix} contains small population computations, detailed calculations for obtaining the reward cost function for $t = 2$, as well as numerical simulations for the punishment and hybrid cost functions.

\section{Model and methods} 
\label{sec: models}

In this section, we present the model and methods of the paper. We first introduce the class of games, namely cooperation dilemmas, that we are interested in throughout this work.
\subsection{Evolutionary processes} 
\label{sec: cost of incentive}
We consider an evolutionary process of a well-mixed, finite population of $N$ interacting individuals (players) and we model the finite population dynamics on an absorbing Markov chain of $(N+1)$ states, $\{S_0, ..., S_N\}$, where $S_j$ represents a population with $j$ cooperators (and $N-j$ defectors) (the sates $S_0$ and $S_N$ are absorbing). We employ the Fermi strategy update rule \cite{traulsen2006} stating that a player $X$ with fitness $f_X$ adopts the strategy of another player $Y$ with fitness $f_Y$ with a probability given by $P_{X,Y}=\left(1 + e^{-\beta(f_Y-f_X)}\right)^{-1}$, where $\beta$ represents the intensity of selection. 

\subsection{Cooperation dilemmas} 

Individuals engage with one another using one of the following one-shot (i.e., non-repeated) cooperation dilemmas: the Donation Game (DG) or its multi-player version, the Public Goods Game (PGG). Strategy wise, each player can choose to either cooperate (C) or  defect (D). 

Let $\Pi_C(j)$ be the average payoff of a C player (cooperator) and $\Pi_D(j)$ that  of a D player (defector), in a population with $j$ $C$ players and $(N-j)$ $D$ players. As can be seen below, the difference in payoffs $\delta = \Pi_C(j) - \Pi_D(j)$ in both games does not depend on $j$. For the two cooperation dilemmas considered in this paper, namely the Donation Games and the Public Goods Games,  it is  always the case that $\delta < 0$. This does not cover some weak social dilemmas such as the snowdrift game, where $\delta>0$ for some $j$, the general prisoner's dilemma, and the collective risk game \cite{sun2021combination}, where $\delta$ depends on $j$. 

\subsubsection*{Donation Game (DG)} 
\noindent The Donation Game is a form of Prisoner's Dilemma in which cooperation corresponds to offering the other player a benefit $B$ at a personal cost $c$, satisfying that $B > c$. Defection means offering nothing. 
The payoff matrix of DG (for the row player) is given as follows  
\[
 \bordermatrix{~ & C & D\cr
                  C & B-c & -c \cr
                  D & B & 0  \cr
                 }. 
\]
%where $c$ and $b$ represent the cost and benefit of cooperation with $b > c$.

\noindent Denoting $\pi_{X,Y}$ the payoff of a strategist  $X$ when playing with a strategist $Y$ from the  payoff matrix above, we obtain
\begin{equation*} 
\begin{split} 
\Pi_C(j) &=\frac{(j-1)\pi_{C,C} + (N-j)\pi_{C,D}}{N-1} = \frac{(j-1) (B-c) + (N-j) (-c)}{N-1}  ,\\
\Pi_D(j) &=\frac{j\pi_{D,C} + (N-j-1)\pi_{D,D}}{N-1} =\frac{j B}{N-1}.
\end{split}
\end{equation*} 
%For ab (popular) parametrised version of the Prisoner's Dilemma game where $T = b$, $R = b-c$, $P = 0$, $S = -c$, we have : 
Thus, 
$$\delta = \Pi_C(j) - \Pi_D(j) =  -\Big(c + \frac{B}{N-1}\Big).$$

\subsubsection*{Public Goods Game (PGG)} 

\noindent In a Public Goods Game, players interact in a group of size $n$, where they decide to cooperate, contributing an amount $c > 0$ to a common pool, or to defect, contributing nothing to the pool. The total contribution in a group is multiplied by a factor $r$, where $1 < r < n$ (for the PGG to be a social dilemma), which is then shared equally among all members of the group, regardless of their strategy. Intuitively, contributing nothing offers one a higher amount of money after redistribution.

The average payoffs, $\Pi_C(j)$ and $\Pi_D(j)$, are calculated based on the assumption that the groups engaging in a public goods game are given by multivariate
hypergeometric sampling. Thereby, for transitions between two pure states, this reduces to sampling, without replacement, from a hypergeometric distribution. More precisely, we obtain \cite{hauert2007}
\begin{equation*} 
\begin{split} 
\Pi_C(j) &= \sum^{n-1}_{i=0}\frac{\dbinom{j-1}{i}\dbinom{N-j}{n-1-i}}{
 \dbinom{N-1}{n-1}} \ \left(\frac{(i+1)rc}{n} - c\right) =\frac{rc}{n}\left(1 + (j-1)\frac{n-1}{N-1}\right) - c ,\\
\Pi_D(j) &=\sum^{n-1}_{i=0}\frac{\dbinom{j}{i}\dbinom{N-1-j}{n-1-i}}{
 \dbinom{N-1}{n-1}} \ \frac{jrc}{n} =\frac{rc(n-1)}{n(N-1)}j.
\end{split}
\end{equation*} 
Thus, 
$$\delta = \Pi_C(j) - \Pi_D(j) = -c \left(1 - \frac{r(N-n)}{n(N-1)} \right).
$$

\subsection{Cost of institutional incentives} 

To reward a cooperator (to punish a defector), the institution has to pay an amount $\theta/a$ ($\theta/b$, respectively) so that the cooperator's (defector's) payoff increases (decreases) by $\theta$, where $a, b > 0$ are constants representing the efficiency ratios of providing this type of incentive. 

In an institutional enforcement setting, we assume that the institution has full information about the population composition or statistics at the time of decision-making. That is, given the well-mixed population setting, we assume that  the number $j$ of cooperators in the population is known. While reward and punishment and their usefulness as institutional incentives have been previously studied, including in \cite{duong2021cost, DuongDurbacHan2022} which are most relevant to this work, the approach in the aforementioned papers is a `full-invest' one, i.e., in each state of the evolutionary process, all cooperators (defectors) are rewarded (punished). 

In the present paper, we consider investment schemes that reward (or/and punish) $t$ $C$ players ($N - t$ $D$ players) whenever the  number of cooperators in the population does not exceed a given threshold $t$ (whenever the number of defectors in the population exceeds a given threshold $t$) for $1 \leq t \leq N-1$. The argument for this type of approach and its applicability in real life is as follows. Firstly, if cooperation is sufficiently frequent, the cooperators might survive by themselves without further need of costly incentives. Secondly, if the institution spreads their incentive budget for too many individuals, then the impact on each individual  might not be enough to alter the global  dynamics. 

Hence, we have, for $1\leq j\leq t$, the  cost per generation for the incentive providing  institution is
\begin{equation}
\label{eq: incentives per generation}
\theta_j = \begin{cases} \frac{j}{a}\theta,\quad \text{reward incentive},\\
\frac{N-j}{b}\theta,\quad \text{punishment incentive},\\
\min\Big(\frac{j}{a}, \frac{N-j}{b}\Big)\theta,\quad\text{mixed incentive},
\end{cases}
\end{equation}
while $\theta_j= 0$ for $t<j\leq N-1$.

Next, we derive the total expected cost of inference over all generations. We do so for two different initial states: i) randomly commencing in the state $S_0$ or the state $S_N$ and ii) starting in the state $S_0$ (as it is more likely to find the population at the homogeneous state of all defectors than that of all cooperators).

We firstly study i), i.e.,  the population is equally likely to start in the homogeneous state $S_0$ (no cooperators) as well as in the homogeneous state $S_N$ (all cooperators).
Let $(n_{ik})_{i,k=1}^{N-1}$ be the entries of the fundamental matrix of the absorbing Markov chain of the evolutionary process. The entries give the expected number of times the population is in the state $S_j$ if it has started in the transient state $S_i$ \cite{kemeny1976finite}. Under the above assumption, a mutant can randomly equally occur either at $S_0$ or $S_N$. Thus, the expected number of visits at state $S_j$ ($1\leq j\leq N-1$) is $\frac{1}{2} (n_{1j} + n_{N-1,j})$. Therefore, the expected cost of inference over all generations is given by
\begin{equation}
\label{eq: incentives (1/2,1/2)}
 E(\theta)=\frac{1}{2}\sum_{j=1}^{N-1} (n_{1j}+n_{N-1,j})\theta_j=\frac{1}{2}\sum_{j=1}^t (n_{1j}+n_{N-1,j})\theta_j,   
\end{equation}
where the second equality follows from the individual-based incentives \eqref{eq: incentives per generation}.
\begin{remark}
%[TA: I THINK WE SHOULD MOVE THIS REMARK TO FUTURE WORKS]
We comment on our assumption that the population is equally likely to start in either homogeneous state. This assumption is reasonable when mutation is negligible and is often made in many works based on agent-based simulations \cite{chen2014optimal,cimpeanu2023does,2ndfreeriding,han2018fostering,carrotstick} (in these works, simulations end whenever the population fixates in a homogeneous state). Our model therefore encapsulates the intermediate-run dynamics, an approximation that is valid if the time-scale is long enough for one type to reach fixation, but too short for the next mutant to appear. It might thus be more practically useful for the optimisation of the institutional budget for providing incentives on an intermediate timescale.   

In the most general case, when mutation is frequent, the above assumption might not be suitable. For example, if cooperators are very likely to fixate in a population of defectors, but defectors are unlikely to fixate in a population of cooperators, mutants are on average more likely to appear in the homogeneous cooperative population (that is in $S_N$). Similarly, if defectors are very likely to fixate in a population of cooperators, but cooperators are unlikely to fixate in a population of defectors, mutants are on average more likely to appear in $S_0$ rather than $S_N$. In general, in the long-run, the population will start at $i = 0$ ($i = N$, respectively) with probability equal to the frequency of D (C) computed at the equilibrium, $f_D = 1/(r+1)$ ($f_C = r/(r+1)$, respectively), where $r = e^{\beta (N-1)(\delta +  \theta)}$. Thus, generally, the expected number of visits at state $S_i$ will be $ f_D n_{1i} + f_C n_{N-1,i}$. The cost function will therefore be
$$
E_r(\theta) = \sum_{j=1}^{N-1} (f_D*n_{1,j} + f_C*n_{N-1,j} )*\theta_j. %general formula for small mutation
$$
We will study the general setting in future work.
\end{remark}

\noindent We now compute the cost function where the dynamics starts from the state $S_0$, with no cooperators, which could be the case in the absence of incentives. Thus, the expected cost of inference over all generations is given by
\begin{equation}
\label{eq: incentives (1,0)}
 E(\theta)=\sum_{j=1}^{N-1} n_{1j}\theta_j=\sum_{j=1}^t n_{1j}\theta_j. 
\end{equation}

\subsubsection{Cooperation frequency and optimal incentives}
Next, we construct the problem of cost optimisation of individual-based institutional incentives (reward, punishment, and hybrid) for maximising the level (or guaranteeing at least a certain level) of cooperative behaviour.

Since the population consists of only two strategies, the fixation  probabilities of a C (D) player in a homogeneous population of D (C) players  when the interference scheme is carried out are, respectively, \cite{nowak}
\begin{equation*} 
%\label{eq:fixprob} 
\begin{split}
\rho_{D,C} &= \left(1+\sum_{i = 1}^{N-1} \prod_{k = 1}^i \frac{1+e^{\beta(\Pi_C(k)-\Pi_D(k) +  \delta(k) \theta)}}{1+e^{-\beta(\Pi_C(k)-\Pi_D(k)+ \delta(k)\theta)}}  \right)^{-1}, \\
\rho_{C,D} &= \left(1+\sum_{i = 1}^{N-1} \prod_{k = 1}^i \frac{1+e^{\beta(\Pi_D(k)-\Pi_C(k) - \delta(k)\theta)}}{1+e^{-\beta(\Pi_D(k)-\Pi_C(k)-\delta(k)\theta)}}  \right)^{-1}.
\end{split}
\end{equation*} 
where $\delta(k) = 1$ if $k \leq t$ and $\delta(k) = 0$ otherwise. 

Computing the stationary distribution using  these fixation probabilities, we  obtain the frequency of cooperation  $$\frac{\rho_{D,C}}{\rho_{D,C}+\rho_{C,D}}.$$

Hence, this frequency of cooperation can be maximised by maximising 
\begin{equation}
\label{eq: max}
\max_{\theta} \left(\rho_{D,C}/\rho_{C,D}\right).  
\end{equation} 

The fraction in Equation~\eqref{eq: max} can be simplified as follows \cite{nowak2006} %\textbf{[TA: I THINK THIS FORMULA NOW DEPENDS ON t -- THIS NEEDS TO BE CORRECTED] -- please check the Equation 4 in my Sci Rep 2018 paper }
\begin{eqnarray}
\label{eq: max_Q_prime}
\nonumber
\frac{\rho_{D,C}}{\rho_{C,D}} &=&  \prod_{k = 1}^{N-1} \frac{u_{i,i-1}}{u_{i,i+1}} =\prod_{k = 1}^{N-1} \frac{1 + e^{\beta[\Pi_C(k)-\Pi_D(k) + \delta(k) \theta]}}{1 + e^{-\beta[\Pi_C(k)-\Pi_D(k) + \delta(k)\theta]}} \\
\nonumber
&=& e^{\beta\sum_{k = 1}^{N-1} \left(\Pi_C(k)-\Pi_D(k) + \delta(k)\theta\right)} \\ \nonumber
%&=& e^{\beta\sum_{k = 1}^{t} \left(\Pi_C(k)-\Pi_D(k) + \theta\right)} \\ 
 &=& e^{\beta [(N-1) \delta +  t \theta]}. 
 \end{eqnarray}

In the above transformation, $u_{i,i-1}$ and $u_{i,i-1}$ are the probabilities  to decrease or increase the number  of C players  (i.e., $i$) by one in each time step, respectively. 

Under neutral selection (i.e., when $\beta = 0$), there is no need to use incentives as no player is likely to copy another player and any changes in strategy that happen are due to noise as opposed to incentives. Thus, we only consider   $\beta > 0$. The goal is  to ensure  at least an $\omega  \in [0,1]$ fraction of cooperation, i.e., $\frac{\rho_{D,C}}{\rho_{D,C}+\rho_{C,D}} \geq \omega$. Thus,  it follows from the equation above  that %\textbf{[TA: THIS NEEDS TO BE CORRECTED] -- please check the Equation 7 in my Sci Rep 2018 paper, for general t }
\begin{equation} 
\label{eq:omega_fraction}
 \theta \geq \theta_0(\omega) = \frac{1}{t}\left( \frac{1}{\beta}\log\left(\frac{\omega}{1-\omega}\right) - (N-1)\delta\right).
\end{equation}
It is guaranteed that, if $\theta  \geq \theta_0(\omega)$, at least an $\omega$ fraction of cooperation is expected in the long run.
This condition implies that the lower bound of $\theta$ monotonically depends on $\beta$. Namely, when $\omega \geq 0.5$, it increases with $\beta$ and when $\omega < 0.5$, it decreases with $\beta$.

Bringing everything together, we obtain the following constrained mathematical minimisation problem of individual-based institutional incentives (reward, punishment, and hybrid) guaranteeing at least a certain level of cooperative behaviour:
\begin{equation}
\label{eq: min prob} \min_{\theta\geq \theta_0} E(\theta),
\end{equation} where $E(\theta)$ can either be the reward ($E_r(\theta)$), punishment ($E_p(\theta)$), or hybrid cost function ($E_{mix}(\theta)$). 

The main aim of this paper is study the above optimisation problem. We focus  on the reward incentive, in which the decision-maker rewards $j$ cooperators whenever $j \leq t$, as in \eqref{eq: incentives per generation} (thus, we set $a=1$ in \eqref{eq: incentives per generation} throughout the paper since it does not affect our optimisation problem). In principle, the punishment and mixed incentives cases are similar, and we only provide numerical investigations for these schemes.  

For $t=1$, we analytically show that $E_r(\theta)$ is non-decreasing as a function of $\theta$ for all values of other parameters. We also establish the neutral and strong selection limits of $E_r(\theta)$, that is
\[
\lim_{\beta\rightarrow 0} E_r(\theta) \quad\text{and}\quad \lim_{\beta\rightarrow \infty} E_r(\theta). 
\]
For other values of $t$, we numerically investigate the properties of $E_r(\theta)$. %These are the contents of the subsequent sections.
\section{Reward incentive under a small mutation limit}
\label{sec: reward}

We first calculate the reward cost function, $E_r(\theta)$, more explicitly. To this end, we compute the expected number of times the population contains $i$ C players for $1 \leq i \leq N-1$. Let  $U = \{u_{ik}\}_{i,k = 1}^{N-1}$ denote the transition matrix between the $N-1$ transient states, $\{S_1, ..., S_{N-1}\}$. According to the Fermi update rule, the transition matrix is given as follows:
\begin{equation} 
\label{eq: transition probabilities reward}
\begin{split} 
u_{i,i\pm k} &= 0 \qquad \text{ for all } k \geq 2, \\
u_{i,i\pm1} &= \frac{N-i}{N} \frac{i}{N} \left(1 + e^{\mp\beta[\Pi_C(i) - \Pi_D(i)+\theta_i/i]}\right)^{-1},\\
u_{i,i} &= 1 - u_{i,i+1} -u_{i,i-1}.
\end{split} 
\end{equation}

For simplicity, we normalise $a=1$, and obtain (recalling that $\Pi_C(i)-\Pi_D (i)=\delta$ for all $1\leq i\leq N-1$, and $\theta_i/i=\frac{\theta}{a}=\theta$ for $1\leq i\leq t$ and zero otherwise):
\begin{align*}
&u_{i,i\pm k}=0 \quad \text{for}\quad k\geq 2,
\\&u_{i,i+1}=\begin{cases}
\frac{(N-i)i}{N^2}\Big(1+e^{-\beta(\delta+\theta)}\Big)^{-1}\quad \text{for}\quad 1\leq i\leq t,\\
\frac{(N-i)i}{N^2}\Big(1+e^{-\beta\delta}\Big)^{-1}\quad \text{for}\quad t< i\leq N-1,\\
\end{cases}
\\&u_{i,i-1}=\begin{cases}
\frac{(N-i)i}{N^2}\Big(1+e^{\beta(\delta+\theta)}\Big)^{-1}\quad \text{for}\quad 1\leq i\leq t,\\
\frac{(N-i)i}{N^2}\Big(1+e^{\beta\delta}\Big)^{-1}\quad \text{for}\quad t< i\leq N-1,\\
\end{cases}\\
u_{i,i}&=1-u_{i,i+1}-u_{i,i-1}.
\end{align*}
Next, we need to calculate the entries $n_{ik}$ of the fundamental matrix $\mathcal{N}=(n_{ik})_{i,k=1}^{N-1}= (I-U)^{-1}$. By using $\frac{1}{1+m}+\frac{1}{1+\frac{1}{m}}=1$ for $m=e^{\beta[\Pi_C(i) - \Pi_D(i)+\theta_i/i]}$, we get $u_{i,i+1}+u_{i,i-1}=\frac{N-i}{N}\frac{i}{N}$. Then, by letting $V=(I-U)$, we obtain: 
\begin{equation} 
\begin{split} 
v_{i,i\pm k} &= 0 \qquad \text{ for all } k \geq 2, \\
v_{i,i\pm1} &= -\frac{N-i}{N} \frac{i}{N} \left(1 + e^{\mp\beta[\Pi_C(i) - \Pi_D(i)+\theta_i/i]}\right)^{-1},\\
v_{i,i} &= u_{i,i+1} +u_{i,i-1}=\frac{N-i}{N}\frac{i}{N}.
\end{split} 
\end{equation}

We can further write $V=W \mathrm{diag}\Big\{\frac{N-1}{N}\frac{1}{N},\ldots,\frac{N-i}{N}\frac{i}{N},\ldots, \frac{1}{N}\frac{N-1}{N}\Big\}$, where 
\begin{equation}
\label{eq: matrix W}
W=\begin{pmatrix}
1&-a&&&&&\\
-c&1&-a&&&&&\\
&\ddots&\ddots&\ddots&&&\\
&&-c&1&-a&&&\\
&&&-d&1&-b&&&\\
&&&&\ddots&\ddots&\ddots&\\
&&&&&-d&1&-b\\
&&&&&&-d&1
\end{pmatrix},
\end{equation}
with $a:=(1+e^{-\beta(\delta+\theta)})^{-1}, \quad b:=(1+e^{-\beta\delta})^{-1}, \quad c:=(1+e^{\beta(\delta+\theta)})^{-1}, \quad d:=(1+e^{\beta\delta})^{-1}$.

This implies  that
$\mathcal{N}=V^{-1}=\mathrm{diag}\Big\{\frac{N^2}{N-1},\frac{N^2}{2(N-2)},\ldots,\frac{N^2}{N-1}\Big\}W^{-1}$,
and so, the fundamental matrix is $\mathcal{N}= \mathrm{diag}\Big\{\frac{N^2}{N-1},\frac{N^2}{2(N-2)},\ldots,\frac{N^2}{N-1}\Big\}(n_{ik})_{i,k=1}^{N-1}= W^{-1}$.

Therefore, the expected total cost of interference for reward in the case of the dynamics starting equally in either state $S_0$ or state $S_N$ is
\begin{align}
\label{eq: total investment reward (1/2,1/2)}
E_r(\theta)&= \frac{1}{2} \sum_{j=1}^{N-1}(n_{1j} + n_{N-1,j}) \theta_j \nonumber
\\&=\frac{1}{2} \sum_{j=1}^{t}(n_{1j} + n_{N-1,j}) \theta_j \nonumber
\\&= \frac{N^2\theta}{2} \sum_{j=1}^{t}\frac{(W^{-1})_{1,j} + (W^{-1})_{N-1,j}}{N-j},
\end{align} where $\theta_j = \theta j$ follows from Equation \eqref{eq: incentives (1/2,1/2)} normalised with $a = 1$.

Similarly, when the population starts in the state $S_0$, we have
\begin{align}
\label{eq: total investment reward (1,0)}
E_r(\theta)&= \sum_{j=1}^{N-1}n_{1j}\theta_j=\sum_{j=1}^{t}n_{1j} \theta_j=\sum_{j=1}^{t}\frac{(W^{-1})_{1,j}}{N-j}.
\end{align}

The advantage of expressing the cost function in terms of the entries of the matrix $W$ is that this matrix is tri-diagonal. The inverse of a tri-diagonal matrix can be theoretically computed using recursive formulae, see e.g., \cite{huang1997}. In general, these formulae are still very hard to analytically explore. However, in the special extreme cases $t=1$ and $t=N-1$, we can explicitly obtain the entries of the inverse matrix $W^{-1}$. Therefore, we obtain analytically the explicit formula for the cost function. The case $t=N-1$ has already been studied in \cite{duong2021cost,DuongDurbacHan2022}. Thus, in this paper, we study the case $t=1$, which will be subsequently discussed in detail.

\subsection{Institutional reward with $t = 1$ for the equally likely starting state}
\label{sec: Reward t = 1 (1/2, 1/2)}

In this section, we introduce the analytical results related to the case of institutional reward with $t = 1$, when the institution provides reward only when there is a \textit{single} cooperator in the population.
We present the cost function for this particular case together with information on its monotonicity, as well as the limits for the neutral drift and strong selection.

The reward cost function for the threshold value $t=1$ is
\begin{equation}
\label{eq: investment reward (1/2,1/2)}
E_r(\theta)= \frac{N^2\theta}{2} \frac{(W^{-1})_{1,1} + (W^{-1})_{N-1,1}}{N-1},
\end{equation} obtained by substituting $t = 1$ in \eqref{eq: total investment reward (1/2,1/2)}. Next, we compute explicitly the entries $(W^{-1})_{1,1}$ and $(W^{-1})_{N-1,1}$. Note that, for $t=1$, the matrix $W$ is a special case of a tri-diagonal matrix of the form
$$
A=\left(\begin{array}{ccccccc}
b_{1} & c_{1} & & & & & \\
a_{2} & b_{2} & c_{2} & & & & \\
& \ddots & \ddots & \ddots & & & \\
& & a_{j} & b_{j} & c_{j} & & \\
& & & \ddots & \ddots & \ddots & \\
& & & & a_{n-1} & b_{n-1} & c_{n-1} \\
& & & & & a_{n} & b_{n}
\end{array}\right) .
$$
We recall the following result from \cite{huang1997} that provides analytical formulae for calculating the entries of the inverse matrix $A^{-1}$. We then apply this result to calculate $(W^{-1})_{1,1}$ and $(W^{-1})_{N-1,1}$.
\begin{theorem}(\cite{huang1997})
\label{thm: diag element}
Define the second-order linear recurrences
$$
z_{i}=b_{i} z_{i-1}-a_{i} c_{i-1} z_{i-2} \quad i=2,3, \ldots, n
$$ where $z_{0}=1, z_{1}=b_{1}$, and
$$
y_{j}=b_{j} y_{j+1}-a_{j+1} c_{j} y_{j+2} \quad j=n-1, n-2, \ldots, 1
$$ where $y_{n+1}=1, y_{n}=b_{n}$. The inverse matrix $A^{-1}=\left\{\phi_{i, j}\right\}(1 \leqslant i, j \leqslant n)$ can be expressed as

$$
\phi_{j, j}=\frac{1}{b_{j}-a_{j} c_{j-1} \frac{z_{j-2}}{z_{j-1}}-a_{j+1} c_{j} \frac{y_{j+2}}{y_{j+1}}}
$$

where $j=1,2, \ldots, n, a_{1}=0, c_{n}=0$ and

$$
\phi_{i, j}= \begin{cases}-c_{i} \frac{z_{i-1}}{z_{i}} \phi_{i+1, j} & i<j \\ -a_{i} \frac{y_{i+1}}{y_{i}} \phi_{i-1, j} & i>j .\end{cases}
$$
\end{theorem}

\begin{corollary}(\cite{huang1997})
\label{cor: fi product}

The inverse matrix $A^{-1}=\left\{\phi_{i, j}\right\}$ can be expressed as

$$
\phi_{j, j}=\frac{1}{b_{j}-a_{j} c_{j-1} \frac{z_{j-2}}{z_{j-1}}-a_{j+1} c_{j} \frac{y_{j+2}}{y_{j+1}}},
$$ where $j=1,2, \ldots, n, \  a_{1}=0$, $c_{n}=0$ and

$$
\phi_{i, j}= \begin{cases}(-1)^{j-i}\left(\prod_{k=1}^{j-i} c_{j-k}\right) \frac{z_{i-1}}{z_{j-1}} \phi_{j, j} & i<j \\ (-1)^{i-j}\left(\prod_{k=1}^{i-j} a_{j+k}\right) \frac{y_{i+1}}{y_{j+1}} \phi_{j, j} & i>j.\end{cases}
$$
\end{corollary}
We use the above results to derive the following lemma which provides an explicit formula for the reward cost function for $t=1$.
\begin{lemma}
\label{lem: first formula of the cost function}
For $t=1$, we have
\begin{equation}
E_r(\theta) = \frac{N^2\theta}{2(N-1)} \Big(1+d^{N-2}\frac{1}{y_2}\Big)\frac{1}{1-ad\frac{y_3}{y_2}},
\end{equation}
where $y_2$ and $y_3$ are found from the following backward recursive formula:
\begin{equation}
\label{eq: recurrence}
y_{N}=1,\quad y_{N-1}=1, \quad \text{and}\quad
y_{N-i} = y_{N-i+1} - (bd)y_{N-i+2},\quad\text{for}\quad i=1,\ldots, N-2.
\end{equation}
\end{lemma}
\begin{proof}
We compute the two entries $W^{-1}_{1,1}$ and $W^{-1}_{N-1,1}$ appearing in \eqref{eq: investment reward (1/2,1/2)} using Theorem \ref{thm: diag element}. We recall that \begin{equation*}
W=\begin{pmatrix}
1&-a&&&&&\\
-c&1&-a&&&&&\\
&\ddots&\ddots&\ddots&&&\\
&&-c&1&-a&&&\\
&&&-d&1&-b&&&\\
&&&&\ddots&\ddots&\ddots&\\
&&&&&-d&1&-b\\
&&&&&&-d&1
\end{pmatrix},
\end{equation*}
with
$$
a:=(1+e^{-\beta(\delta+\theta)})^{-1}, \quad b:=(1+e^{-\beta\delta})^{-1}, \quad c:=(1+e^{\beta(\delta+\theta)})^{-1}, \quad d:=(1+e^{\beta\delta})^{-1}.
$$
In particular, for $t=1$, we obtain
\begin{equation*}
W=\begin{pmatrix}
1&-a&&&&&\\
-d&1&-b&&&&&\\
&\ddots&\ddots&\ddots&&&\\
&&-d&1&-b&&&\\
&&&-d&1&-b&&&\\
&&&&\ddots&\ddots&\ddots&\\
&&&&&-d&1&-b\\
&&&&&&-d&1
\end{pmatrix}.
\end{equation*}
%= \begin{pmatrix}
%A&B\\
%C&D\\
%\end{pmatrix},
%\end{equation*}
%where $A = 1 \in \mathbb{R}^{1\times 1}$, $B =\begin{bmatrix} -a & 0 & \ldots & 0 \end{bmatrix}\in \mathbb{R}^{1\times (N-2)}$, $C = \begin{bmatrix}
%           -d \\
%           0 \\
%           \vdots \\
%           0
%         \end{bmatrix}\in \mathbb{R}^{(N-2)\times 1}$, $D\in \mathbb{R}^{(N-2)\times (N-2)}$.
 
We apply Theorem \ref{thm: diag element}, the diagonal element case, and Corollary \ref{cor: fi product}, the $i>j$ case, to obtain 
\begin{align*}
W^{-1}_{1,1}&= \frac{1}{1-ad\frac{y_3}{y_2}},\\
(W^{-1})_{N-1,1} &= (-1)^{N-2}\Big(\prod\limits_{k=1}^{N-2} a_{k+1}\Big) \frac{y_N}{y_2}\phi_{1,1}
\\&= (-1)^{N-2}\Bigg(\prod\limits_{k=1}^{N-2} a_{k+1}\Bigg)\frac{1}{y_2}\frac{1}{1-ad\frac{y_3}{y_2}}
\\&= d^{N-2}\frac{1}{y_2}\frac{1}{1-ad\frac{y_3}{y_2}}.
\end{align*}
In the above formulae, $y_2$ and $y_3$ are found from the backward recursive formula \eqref{eq: recurrence}. Thus, by summing up the above expressions, we obtain
\begin{align*}
(W^{-1})_{1,1} + (W^{-1})_{N-1,1} &= \frac{1}{1-ad\frac{y_3}{y_2}} + d^{N-2}\frac{1}{y_2}\frac{1}{1-ad\frac{y_3}{y_2}}
\\& = \Big(1+d^{N-2}\frac{1}{y_2}\Big)\frac{1}{1-ad\frac{y_3}{y_2}}.
\end{align*}
Substituting the above into \eqref{eq: investment reward (1/2,1/2)}, we get
\begin{align*}
E_r(\theta)&= \frac{N^2\theta}{2(N-1)} \Big(1+d^{N-2}\frac{1}{y_2}\Big)\frac{1}{1-ad\frac{y_3}{y_2}},
\end{align*}
which completes the proof of this lemma.
 \end{proof} 
We now calculate  explicitly $y_2$ and $y_3$ using the backward recursive formula \eqref{eq: recurrence}:
\begin{align*}
& \left\{\begin{array}{l}
y_{N-1}=1, \hspace{1mm} y_{N}=1 \\
y_{j}=y_{j+1}-(b d) y_{j+2} \hspace{1mm} \text{for} \hspace{1mm} j=N-2, \ldots, 2.\\
\end{array} \right. 
\end{align*}
For convenience, we transform the above backward recursive relation to a forward one. By applying a change of variable $\hat{y}_{j}=y_{N-j}$, we get
\begin{align*}
& \left\{\begin{array}{l}
\hat{y}_{0}=1, \hspace{1mm} \hat{y}_{1}=1 . \\
\hat{y}_{j}=\hat{y}_{j-1}-(b d) \hat{y}_{j-2} \hspace{1mm} \text{for} \hspace{1mm} j=2,\ldots, N .
\end{array}\right.
\end{align*}
By induction, it follows that $\hat{y}_j$ can be written in the form
\[
\hat{y}_j=\sum_{k=0}^j m_k (bd)^k.
\]
We observe that the coefficients of the recurrence relation follow the pattern in the table below.
\begin{table}[H]
\centering
\begin{tabular}{llllllll}
 &  & $(bd)^0$ &  $(bd)^1$ & $(bd)^2$  & $(bd)^3$ & $(bd)^4$ & $(bd)^5$  \\
$\hat{y}_0$ &  & 1 & 0 & 0 & 0 & 0 & 0 \\
$\hat{y}_1$ &  & 1 & 0 & 0 & 0 & 0 & 0 \\
$\hat{y}_2$ &  & 1 & -1 & 0 & 0 & 0 & 0 \\
$\hat{y}_3$ &  & 1 & -2 & 0 & 0 & 0 & 0 \\
$\hat{y}_4$ &  & 1 & -3 & 1 & 0 & 0 & 0 \\
$\hat{y}_5$ &  & 1 & -4 & 3 & 0 & 0 & 0 \\
$\hat{y}_6$ &  & 1 & -5 & 6 & -1 & 0 & 0 \\
$\hat{y}_7$ &  & 1 & -6 & 10 & -4 & 0 & 0 \\
$\hat{y}_8$ &  & 1 & -7 & 15 & -10 & 1 & 0 \\
$\hat{y}_9$ &  & 1 & -8 & 21 & -20 & 5 & 0
\end{tabular}
\caption{Pattern for the coefficients in the expansion of $\hat{y}_j$ for $j=0,1\,\ldots$, in terms of the power of $(bd)$.}
\label{Tab: binomial coefficients}
\end{table}
We observe that the entries of Table \ref{Tab: binomial coefficients} are the binomial coefficients offset by a factor of $k$ for every column and are alternating in sign. This suggests that $\hat{y}_j = \sum\limits_{k=0}^{j} (-1)^k {j-k \choose k}(bd)^k$. In the lemma below, we prove this is indeed the case. We also obtain another expression for $\hat{y}_j$ using the general formula of a second order homogeneous recurrence relation. 
\begin{lemma}[Recurrence relation]\
\label{lem: recurrence relation}
The following formulae hold for any $j\geq 0$
\begin{enumerate}[(1)]
       \item $\hat{y}_j= \sum\limits_{k=0}^{\lfloor\frac{j}{2}\rfloor} (-1)^k {j-k \choose k} (bd)^k$,
       \item $\hat{y}_j = \frac{x_{1} x_{2}\left(x_{1}^{j-1}-x_{2}^{j-1}\right)-\left(x_{1}^{j}-x_{2}^{j}\right)}{x_{2}-x_{1}}$,
\end{enumerate}
where
$$
x_{1,2}=\frac{1 \pm \sqrt{1-4 b d}}{2}.
$$
\end{lemma}
\begin{proof}
We prove the first statement by induction on $j$. For $j = 2$ we have $\hat{y}_2 = \hat{y}_1-(b d) \hat{y}_0 = 1 - bd$. We assume the statement holds for $j$ and need to prove it for $j+1$. In fact, we have
\begin{align}
\label{eq: y induction}
\hat{y}_{j+1} & =\hat{y}_{j}-(b d) \hat{y}_{j-1} \nonumber
\\& =\sum_{k=0}^{\lfloor\frac{j}{2}\rfloor} (-1)^k {j-k \choose k}(b d)^{k}-(b d) \sum_{k=0}^{\lfloor\frac{j-1}{2}\rfloor}(-1)^k {j-1-k \choose k}(b d)^{k} \nonumber
\\& =\sum_{k=0}^{\lfloor\frac{j}{2}\rfloor}(-1)^k {j-k \choose k}(b d)^{k}-\Bigg(\sum_{k=0}^{\lfloor\frac{j-1}{2}\rfloor} (-1)^k {j-1-k \choose k}(b d)^{k+1}\Bigg) \nonumber
\\& =\sum_{k=0}^{\lfloor\frac{j}{2}\rfloor} (-1)^k {j-k \choose k}(b d)^{k}-\sum_{k=1}^{\lfloor\frac{j+1}{2}\rfloor} (-1)^{k-1} {j-k \choose k-1}(b d)^{k} \nonumber
\\& = 1 +\sum_{k=1}^{\lfloor\frac{j}{2}\rfloor}\Bigg((-1)^k {j-k \choose k}-(-1)^{k-1} {j-k \choose k-1}\Bigg)(b d)^{k} - (-1)^{\lfloor \frac{j+1}{2}\rfloor - 1} {j- \lfloor \frac{j+1}{2}\rfloor \choose \lfloor \frac{j+1}{2}\rfloor - 1} (bd)^{\lfloor \frac{j+1}{2}\rfloor} \nonumber \\ 
& =  1 +\sum_{k=1}^{\lfloor\frac{j}{2}\rfloor}(-1)^k \Bigg({j-k \choose k} + {j-k \choose k-1}\Bigg)(b d)^{k} + (-1)^{\lfloor \frac{j+1}{2}\rfloor} {j- \lfloor \frac{j+1}{2}\rfloor \choose \lfloor \frac{j+1}{2}\rfloor - 1} (bd)^{\lfloor \frac{j+1}{2}\rfloor} \\
& = 1 +\sum_{k=1}^{\lfloor\frac{j}{2}\rfloor}(-1)^k {j+1-k \choose k}(b d)^{k} + (-1)^{\lfloor \frac{j+1}{2}\rfloor} {j- \lfloor \frac{j+1}{2}\rfloor \choose \lfloor \frac{j+1}{2}\rfloor - 1} (bd)^{\lfloor \frac{j+1}{2}\rfloor} \nonumber\\
& =\sum_{k=0}^{\lfloor\frac{j+1}{2}\rfloor} (-1)^k {j+1-k \choose k}(b d)^{k} \nonumber, 
\end{align} where, to obtain \eqref{eq: y induction}, we used Pascal's identity 
\[ {j-k \choose k} + {j-k \choose k-1} = {j+1-k \choose k} \quad \text{for} \quad 1\leq k \leq \lfloor\frac{j}{2}\rfloor
\]
and 
\[
{j+1-\lfloor \frac{j+1}{2}\rfloor \choose \lfloor \frac{j+1}{2}\rfloor} = {j - \lfloor \frac{j+1}{2}\rfloor \choose \lfloor \frac{j+1}{2}\rfloor} + {j-\lfloor \frac{j+1}{2}\rfloor \choose \lfloor \frac{j+1}{2}\rfloor - 1},
\]

\noindent with ${j-\lfloor \frac{j+1}{2}\rfloor \choose \lfloor \frac{j+1}{2}\rfloor - 1} = 0$ as $\lfloor \frac{j+1}{2}\rfloor - 1 > j-\lfloor \frac{j+1}{2}\rfloor$. 

\vspace{1mm}

\noindent Moreover, $\lfloor \frac{j-1}{2} \rfloor + 1 = n+1= \lfloor \frac{j+1}{2} \rfloor$ if $j=2n+1$ and $\lfloor \frac{j-1}{2} \rfloor + 1 = n = \lfloor \frac{j+1}{2} \rfloor$ if $j=2n$.

\vspace{0.5cm}

\noindent For the proof of the second statement, we note that $\hat{y}_{j}=\hat{y}_{j-1}-(b d) \hat{y}_{j-2}$ is a second order homogeneous recurrence relation and we solve it using the characteristic equation 
$$
x^{2}-x+b d=0.
$$
Recall $b=\frac{1}{1+e^{-\beta \delta}}$ and $d=\frac{1}{1+e^{\beta \delta}}$, so 
\begin{align*}
b d &=\frac{1}{\left(1+e^{\beta \delta}\right)\left(1+e^{-\beta \delta}\right)}
\\&=\frac{1}{2+e^{\beta \delta}+e^{-\beta \delta}}.
\end{align*}
Since $e^{\beta \delta}+e^{-\beta \delta} \geq 2 \sqrt{e^{\beta \delta} e^{-\beta \delta}}=2$, $b d < \frac{1}{4}$ $(\text{as }\beta, \delta \neq 0)$.\\
Therefore, the characteristic equation has two real solutions 
$$
x_{1,2}=\frac{1 \pm \sqrt{1-4 b d}}{2}.
$$
Thus, we can express $\hat{y}_i$ using these roots as
\[
\hat{y}_{j}=r_{1} x_{1}^{j}+r_{2} x_{2}^{j},
\]
where the two constants are found from the initial data:
$$
\hat{y}_{0}=r_{1}+r_{2}=1 \quad\text{and}\quad\hat{y}_{1}=r_{1} x_{1}+r_{2} x_{2}=1.
$$
Solving the above system for $r_1, r_2$ results in 
\[
r_{1}=\frac{x_{2}-1}{x_{2}-x_{1}},\quad r_{2}=\frac{1-x_{1}}{x_{2}-x_{1}}.
\]
Hence
\[
\begin{aligned}
\hat{y}_{j} & =\frac{x_{2}-1}{x_{2}-x_{1}} x_{1}^{j}-\frac{x_{1}-1}{x_{2}-x_{1}} x_{2}^{j} \\
& =\frac{x_{1}^{j} x_{2}-x_{1} x_{2}^{j}-\left(x_{1}^{j}-x_{2}^{j}\right)}{\left(x_{2}-x_{1}\right)} \\
& =\frac{x_{1} x_{2}\left(x_{1}^{j-1}-x_{2}^{j-1}\right)-\left(x_{1}^{j}-x_{2}^{j}\right)}{x_{2}-x_{1}}.
\end{aligned} 
\]
This completes the proof of the lemma.
\end{proof}
 The following theorem, which is the main analytical result of the present paper, provides an explicit formula for the reward cost function $E_r(\theta)$ and shows that it is always non-decreasing for all parameter values. 
\begin{theorem}[Derivative and monotonicity of the cost function]\
\label{thm: derivative positive t1 (1/2, 1/2)}

\noindent $E_r'(\theta)$ is always increasing with respect to $\theta$ for all values of $N,\theta,\beta$, where
\[
E_r'(\theta)=\frac{N^2}{(2N-2)}\Big(y_2+d^{N-2}\Big)\Big(\frac{1}{y_2-ady_3}\Big)\Big[1+\frac{\theta dy_3\beta e^{-\beta(\delta+\theta) }a^2}{y_2-ad y_3}\Big].
\]
As a consequence, the minimisation problem \eqref{eq: min prob} has a unique solution
\[
\min\limits_{\theta\geq \theta_0} E_r(\theta)= E_r(\theta_0).
\]
\end{theorem}

\begin{proof}
Recalling from Lemma \ref{lem: first formula of the cost function} that
\[
E_r(\theta)= \frac{N^2\theta}{2N-2}(y_2+d^{N-2})\Big(\frac{1}{y_2-ady_3}\Big).
\]
Next, we compute the derivative of $E_r(\theta)$ with respect to $\theta$. 

\noindent Note that only $a$ depends on $\theta$
\[
a=a(\theta)=(1+e^{-\beta(\theta+\delta)})^{-1}, \quad a'(\theta)=\frac{\beta e^{-\beta(\delta+\theta)}}{(1+e^{-\beta(\theta+\delta})^2} =\beta e^{-\beta(\delta+\theta)} a^2,
\]
while $b=(1+e^{-\beta \delta})^{-1}$ and $d=(1+e^{\beta \delta})^{-1}$ (and thus $y_2$ and $y_3$) do not depend on $\theta$.\\ 
Let
\[
C=C(\theta)=\frac{1}{y_2-ady_3}.
\]
Then the derivative of $C$ with respect to $\theta$ is given by
\begin{align*}
C'(\theta) &=\frac{dy_3 a'(\theta)}{(y_2-ad y_3)^2}\nonumber
\\&=\frac{d y_3 \beta e^{-\beta(\delta+\theta)} a^2}{(y_2-ad y_3)^2}\nonumber
\\&=\frac{d y_3 \beta e^{-\beta(\delta+\theta)} a^2}{(y_2-ad y_3)}C(\theta)\nonumber\\
      &= \frac{y_{3} \beta  e^{- \beta \left(\delta + \theta\right)}}{\left(1 + e^{- \beta \left(\delta + \theta\right)}\right)^{2} \left(y_{2} - \frac{y_{3}}{\left(1 + e^{- \beta \left(\delta + \theta\right)}\right) \left(1+ e^{\beta \delta}\right)}\right)^{2} \left(1 + e^{\beta \delta}\right)}.
\end{align*}
We calculate $E'_r(\theta)$ via the product rule:
\begin{align*}
    E_r'(\theta)&=\frac{N^2}{2N-2}(y_2+d^{N-2})\Big[C(\theta)+\theta C'(\theta)\Big]\\
    &=\frac{N^2}{2N-2}(y_2+d^{N-2})\Big[1+\frac{\theta dy_3\beta e^{-\beta(\delta+\theta) }a^2}{y_2-ad y_3}\Big]C(\theta)
    \\&= \frac{N^2}{(2N-2)}\Big(y_2+d^{N-2}\Big)\Big(\frac{1}{y_2-ady_3}\Big)\Big[1+\frac{\theta dy_3\beta e^{-\beta(\delta+\theta) }a^2}{y_2-ad y_3}\Big].
\end{align*}
From the formula of $E_r(\theta)$, it follows that $y_2-ad y_3>0$.\\

\noindent We next show that $y_{2}>0$ and $y_{3}>0$. Recalling that $y_2=\hat{y}_{N-1}$ and $y_3=\hat{y}_{N-2}$, where the sequence of numbers $(\hat{y})_{j}$ is given in Lemma \ref{lem: recurrence relation}. We show that $\hat{y}_j>0$ for all $j\geq 0$.

%\vspace{4mm}

%Now, by induction, we prove that $0<\hat{y}_{j} < 1$ for all $j\geq 2$ in order to conclude $E^{\prime}(\theta)>0$.

%\vspace{1mm}

%We test the statement for $j=2, 3$. Recall that $0<bd\leq\frac{1}{4}$:

%$$
%\begin{aligned}
%& \hat{y}_{2}=1-(b d), \text{ so } \hat{y}_{2} \in\left(\frac{3}{4}, 1\right),\\
%& \hat{y}_{3}=1-2(b d), \text{ so } \hat{y}_{3} \in\left(\frac{1}{2}, 1\right). 
%\end{aligned}
%$$

%We assume that the statement holds up to $j$, i.e., $\hat{y}_{2}, \hat{y}_{3}, \cdots, \hat{y}_{j} < 1$, and prove it for $j+1$.

%\vspace{1mm}

%By the recurrence relation, $\hat{y}_{j+1}=\hat{y}_{j}-(b d)y_{j-1}$. Since $0<\hat{y}_{j-1}, \hat{y}_j<1$ and $bd<\frac{1}{4}$, then we get that $\hat{y}_{j+1}=\hat{y}_{j}-(b d)y_{j-1}<\hat{y}_j<1$ via the inductive step.

%\vspace{1mm}
In fact, according to the second statement of Lemma \ref{lem: recurrence relation}, we have
\begin{align*}
    \hat{y}_j & = \frac{x_{1} x_{2}\left(x_{1}^{j-1}-x_{2}^{j-1}\right)-\left(x_{1}^{j}-x_{2}^{j}\right)}{x_{2}-x_{1}}
    \\& = \frac{x_1x_2(x_1-x_2)\sum\limits_{k=0}^{j-2} x_1^kx_2^{j-2-k}-(x_1-x_2)\sum\limits_{k=0}^{j-1} x_1^kx_2^{j-1-k}}{x_2-x_1}
    \\& = \sum\limits_{k=0}^{j-1} x_1^kx_2^{j-1-k} - \sum\limits_{k=0}^{j-2} x_1^{k+1}x_2^{j-1-k}
    \\& = \sum\limits_{k=0}^{j-2} x_1^kx_2^{j-1-k} + x_1^{j-1}x_2^0 - \sum\limits_{k=0}^{j-2} x_1^{k+1}x_2^{j-1-k}
    \\& = \sum\limits_{k=0}^{j-2} (x_1^kx_2^{j-1-k} - x_1^{k+1}x_2^{j-1-k}) + x_1^{j-1}
    \\& = \sum\limits_{k=0}^{j-2} x_1^kx_2^{j-1-k}(1-x_1) + x_1^{j-1} >0,
\end{align*} as $x_1,x_2 \in (0,1)$.
Hence $y_2, y_3>0$. Therefore, $E_r'(\theta)>0$ for all $\theta>0$.
\end{proof}

\subsection{Asymptotic limits}
We now study the neutral drift and strong selection limits of the reward cost function $E_r(\theta)$ with $t=1$ for the equally likely initial state when the intensity of selection $\beta$ tends to $0$ and to +$\infty$, respectively.
\begin{proposition}(Neutral drift limit) It holds that
\label{prop: neutral drift t=1 (1/2, 1/2)}    
$$
\lim\limits_{\beta\rightarrow 0}E_r(\theta)=\frac{N^{2} \theta}{N-1} \Bigg(P(N)+\frac{1}{2^{N-2}}\Bigg)\Big(\frac{2}{4P(N)-Q(N)}\Big),
$$
where
\begin{align*}
P(N)&= 1+(-1)^1{N-2-1 \choose 1}\Big(\frac{1}{4}\Big)^1+(-1)^2{N-2-2 \choose 2}\Big(\frac{1}{4}\Big)^{2}+\ldots +(-1)^{\left\lfloor\frac{N-2}{2}\right\rfloor}{N-2-\lfloor\frac{N-2}{2}\rfloor \choose \lfloor\frac{N-2}{2}\rfloor}\Big(\frac{1}{4}\Big)^{\lfloor\frac{N-2}{2}\rfloor},\\
Q(N)& = 1+(-1)^{1}{N-3-1 \choose 1}\Big(\frac{1}{4}\Big)^1+(-1)^{2}{N-3-2 \choose 2}\Big(\frac{1}{4}\Big)^{2}+\ldots + (-1)^{\left\lfloor\frac{N-3}{2}\right\rfloor}{N-3-\lfloor\frac{N-3}{2}\rfloor \choose \lfloor\frac{N-3}{2}\rfloor}\Big(\frac{1}{4}\Big)^{\lfloor\frac{N-3}{2}\rfloor}.
\end{align*}
\end{proposition}
\begin{proof}
Recall from Lemma \ref{lem: first formula of the cost function} that
\[
E_r(\theta)= \frac{N^2\theta}{2N-2}(y_2+d^{N-2})\Big(\frac{1}{y_2-ady_3}\Big)=\frac{N^2\theta}{2N-2}\Big(1+\frac{d^{N-2}}{y_2}\Big)\Big(\frac{1}{1-ad\frac{y_3}{y_2}}\Big).
\]
According to the first statement of Lemma \ref{lem: recurrence relation}
\begin{align*}
 y_2&=\hat{y}_{N-2}= \sum\limits_{k=0}^{\lfloor\frac{N-2}{2}\rfloor}(-1)^{k}{N-2-k \choose k}(b d)^{k},\\
 y_3&=\hat{y}_{N-3}=\sum\limits_{k=0}^{\lfloor\frac{N-3}{2}\rfloor}(-1)^{k}{N-3-k \choose k}(b d)^{k}.
\end{align*}
Note that $a$, $b$, and $d$ depend on $\beta$, and therefore $y_2$ and $y_3$ also depend on $\beta$ through the product $(bd)$. We have
\begin{align*}
&\lim\limits_{\beta\rightarrow 0} d^{N-2}=\lim\limits_{\beta\rightarrow 0}\left(\frac{1}{1+e^{\beta \delta}}\right)^{N-2}=\frac{1}{2^{N-2}},\\
  &\lim\limits_{\beta\rightarrow 0} (bd)= \lim\limits_{\beta\rightarrow 0}\frac{1}{2+e^{\beta \delta}+e^{-\beta \delta}}= \frac{1}{4},\\
&\lim\limits_{\beta\rightarrow 0} (ad)=\lim\limits_{\beta\rightarrow 0} \frac{1}{(1+e^{\beta \delta})(1+e^{-\beta(\theta+\delta)})}=\frac{1}{4}.
\end{align*}
It also follows that
\[
\lim\limits_{\beta\rightarrow 0} y_2= P(N), \quad \lim\limits_{\beta\rightarrow 0} y_3= Q(N),
\]
where $P(N)$ and $Q(N)$ are given explicitly in the statement of the Proposition.
Putting everything together yields 
\begin{align*}
\lim _{\beta \rightarrow 0} E_{r}(\theta)&=\lim _{\beta \rightarrow 0} \Bigg(\frac{N^2\theta}{2N-2}\Big(1+\frac{d^{N-2}}{y_2}\Big)\Bigg(\frac{1}{1-ad\frac{y_3}{y_2}}\Bigg)\Bigg)
\\&=\frac{N^{2} \theta}{2(N-1)} \Bigg(1+\frac{\frac{1}{2^{N-2}}}{P(N)}\Bigg)\left(\frac{1 }{1-\frac{1}{4}\frac{Q(N)}{P(N)}}\right)\\
&= \frac{N^{2} \theta}{N-1} \Bigg(P(N)+\frac{1}{2^{N-2}}\Bigg)\Big(\frac{2}{4P(N)-Q(N)}\Big).
\end{align*}

\end{proof}

\begin{figure}[H]
\centering
\hspace*{-2cm}  
\includegraphics[width=0.7\textwidth]{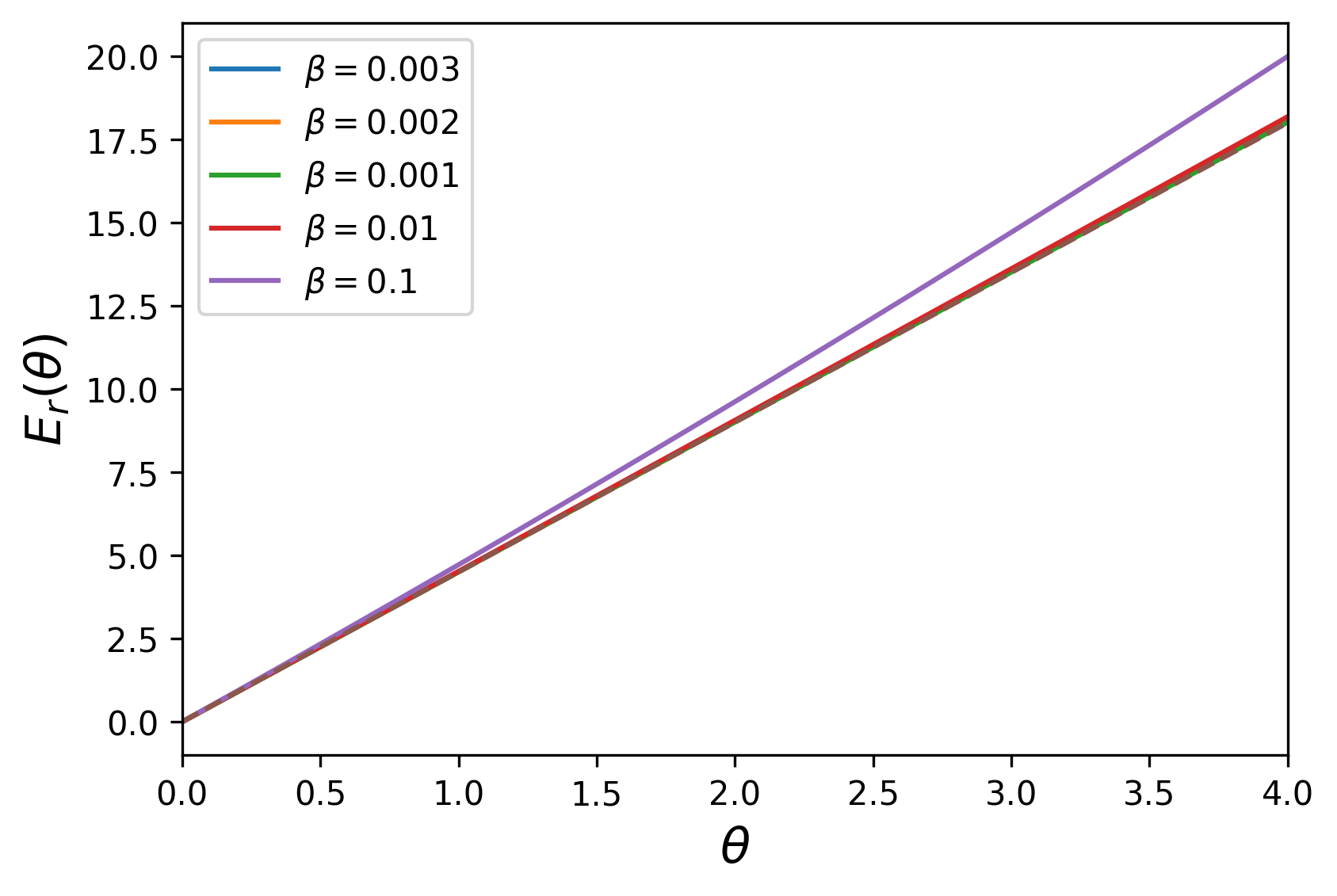}
\caption{The neutral drift limit (dashed brown line) for the reward cost function $E_{r}(\theta)$ with $N = 3$, $t = 1$, and $\theta = 1$ for DG with $B = 2, c = 1$. We notice that, the more the $\beta$ value approaches 0, the closer to the limiting value the cost function $E_{r}(\theta)$ gets, in accordance to the analytical result Proposition \ref{prop: neutral drift t=1 (1/2, 1/2)}.} 
\label{fig:neutraldriftDG1/2}
\end{figure}

\begin{proposition}(Strong selection limit) It holds that
\label{prop: infinite selection t=1 (1/2, 1/2)}    
$$
\lim\limits_{\beta\rightarrow \infty}E_r(\theta)=\begin{cases}
+\infty \quad\text{if}\quad \delta+\theta> 0,\\
\frac{2 N^2\theta}{N-1}\quad\text{if}\quad \delta+\theta= 0,\\
\frac{N^2\theta}{N-1} \quad\text{if}\quad \delta+\theta< 0.
\end{cases}
$$
\end{proposition}
\begin{proof}
We proceed as in the proof of the previous proposition by computing the limit of relevant quantities as $\beta\rightarrow \infty$ instead of $\beta\rightarrow 0$, noting that $\delta<0$ in our setting. We have
\begin{align*}
&\lim\limits_{\beta\rightarrow \infty} d^{N-2}=\lim\limits_{\beta\rightarrow \infty}\left(\frac{1}{1+e^{\beta \delta}}\right)^{N-2}=1,\\
  &\lim\limits_{\beta\rightarrow \infty} (bd)= \lim\limits_{\beta\rightarrow \infty}\frac{1}{2+e^{\beta \delta}+e^{-\beta \delta}}= 0,\\
&\lim\limits_{\beta\rightarrow \infty} (ad)=\lim\limits_{\beta\rightarrow \infty} \frac{1}{(1+e^{\beta \delta})(1+e^{-\beta(\theta+\delta)})}=\begin{cases}
1 \quad\text{if}\quad \delta+\theta> 0,\\
\frac{1}{2}\quad\text{if}\quad \delta+\theta= 0,\\
0 \quad\text{if}\quad \delta+\theta< 0.
\end{cases}
\end{align*}
It also follows that
\[
\lim\limits_{\beta\rightarrow \infty} y_2= 1, \quad \lim\limits_{\beta\rightarrow \infty} y_3=1.
\]
Putting everything together yields 
\begin{align*}
\lim _{\beta \rightarrow 0} E_{r}(\theta)&=\lim _{\beta \rightarrow 0} \Bigg(\frac{N^2\theta}{2N-2}(1+\frac{d^{N-2}}{y_2})\Big(\frac{1}{1-ad\frac{y_3}{y_2}}\Big)\Bigg)
\\&=\begin{cases}
+\infty \quad\text{if}\quad \delta+\theta> 0,\\
\frac{2 N^2\theta}{N-1}\quad\text{if}\quad \delta+\theta= 0,\\
\frac{N^2\theta}{N-1} \quad\text{if}\quad \delta+\theta< 0.
\end{cases}
\end{align*}

\end{proof}

\begin{figure}[H]
\centering
\hspace*{-2cm}  
\includegraphics[width=1.2\textwidth]{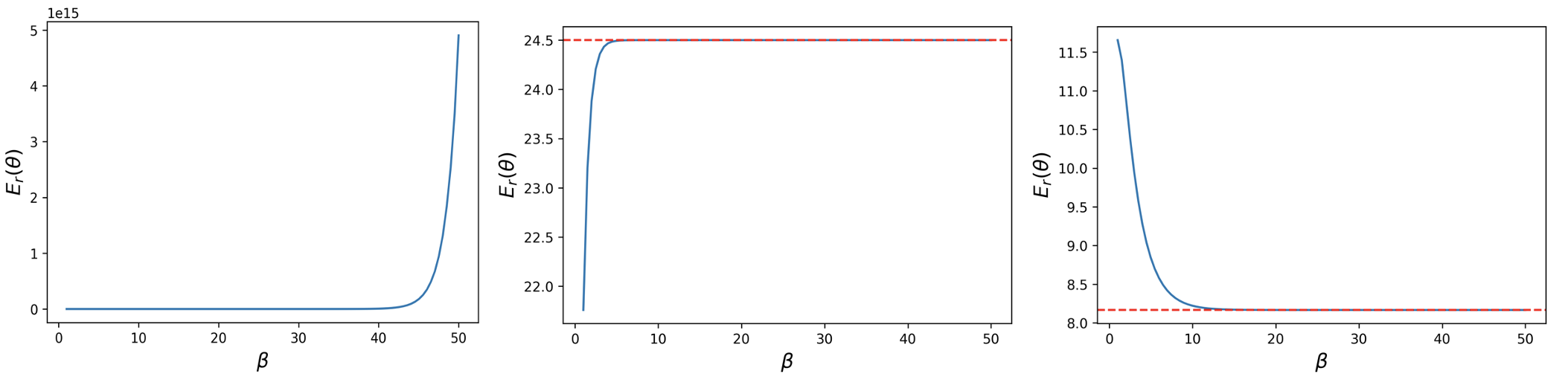}
\caption{The strong selection limit (dashed red line) for the reward cost function $E_{r}(\theta)$ with $N = 7$, $t = 1$ for DG with $B = 2, c = 1$. The first image corresponds to $\delta + \theta >0$ ($\delta = -0.5$, $\theta = 1$), the middle one to $\delta + \theta = 0$ ($\delta = - 1$, $\theta = 1$), and the last one to $\delta + \theta < 0$ ($\delta$ = -1.5, $\theta = 1$). The numerical results are in accordance with the analytical ones in Proposition \ref{prop: infinite selection t=1 (1/2, 1/2)}.} 
\label{fig:strongselectionDG1/2}
\end{figure}

\subsection{Institutional reward with $t = 1$ for the all-defector starting state}
\label{sec: Reward t = 1 (1, 0)}

In this section, we introduce the analytical results related to the case of institutional reward with $t = 1$, when the institution provides reward only when there is a \textit{single} cooperator in the population. Differently from Section \ref{sec: Reward t = 1 (1/2, 1/2)}, here we assume that the population starts in the state $S_0$ of no cooperators (all defectors). We present the cost function for this particular case together with information on its monotonicity, as well as the limits for the neutral drift and strong selection.

The reward cost function for the threshold value $t=1$ is
\begin{equation}
\label{eq: investment reward (1,0)}
E_r(\theta)= N^2\theta \frac{(W^{-1})_{1,1}}{N-1},
\end{equation} which is obtained by substituting $t = 1$ in \eqref{eq: total investment reward (1,0)}. The computation of the entry $(W^{-1})_{1,1}$ is identical to the one in Section \ref{sec: Reward t = 1 (1/2, 1/2)}, that is $(W^{-1})_{1,1} = \Big({1-ad\frac{y_3}{y_2}}\Big)$. 
Thus, 
$$
E_r(\theta) = \frac{N^2\theta}{N-1} \Big(\frac{1}{1-ad\frac{y_3}{y_2}}\Big).
$$
The following theorem, which is the counterpart of Theorem \ref{thm: derivative positive t1 (1/2, 1/2)} in Section \ref{sec: Reward t = 1 (1/2, 1/2)}, provides an explicit formula for the reward cost function $E_r(\theta)$ and shows that it is always non-decreasing for all parameter values. 
\begin{theorem}[Derivative and monotonicity of the cost function]\
\label{thm: derivative positive t1 (1, 0)}

\noindent $E_r'(\theta)$ is always increasing with respect to $\theta$ for all values of $N$, $\theta$ and $\beta$, where
\[
E_r'(\theta)=\frac{N^2}{N-1}\Big(\frac{y_2}{y_2-ady_3}\Big)\Big[1+\frac{\theta dy_3\beta e^{-\beta(\delta+\theta) }a^2}{y_2-ad y_3}\Big].
\]
As a consequence, the minimisation problem \eqref{eq: min prob} has a unique solution
\[
\min\limits_{\theta\geq \theta_0} E_r(\theta)= E_r(\theta_0).
\]
\end{theorem}

\begin{proof}
We proceed like in Theorem \ref{thm: derivative positive t1 (1/2, 1/2)}, recalling from Lemma \ref{lem: first formula of the cost function} that
\[
E_r(\theta) = \frac{N^2\theta}{N-1} \Big(\frac{1}{1-ad\frac{y_3}{y_2}}\Big).
\]    
\noindent Noting, as before, that only $a$ depends on $\theta$, while $b$ and $d$ (and thus $y_2$ and $y_3$) do not, we let $C=C(\theta)=\frac{y_2}{y_2-ady_3}$ and compute its derivative 
$$
C'(\theta) = \frac{y_{3} \beta  e^{- \beta \left(\delta + \theta\right)}}{\left(1 + e^{- \beta \left(\delta + \theta\right)}\right)^{2} \left(y_{2} - \frac{y_{3}}{\left(1 + e^{- \beta \left(\delta + \theta\right)}\right) \left(1+ e^{\beta \delta}\right)}\right)^{2} \left(1 + e^{\beta \delta}\right)}.
$$
We then calculate $E'_r(\theta)$ via the product rule:
\begin{align*}
    E_r'(\theta)&=\frac{N^2}{N-1}\Big[C(\theta)+\theta C'(\theta)\Big]\\
    &=\frac{N^2}{N-1}C(\theta)\Big[1+\frac{\theta dy_3\beta e^{-\beta(\delta+\theta) }a^2}{y_2-ad y_3}\Big]
    \\&= \frac{N^2}{N-1}\Big(\frac{y_2}{y_2-ady_3}\Big)\Big[1+\frac{\theta dy_3\beta e^{-\beta(\delta+\theta) }a^2}{y_2-ad y_3}\Big].
\end{align*}

\noindent Finally, the proof of $E_r'(\theta)>0$ for all $\theta>0$ follows closely that of Theorem \ref{thm: derivative positive t1 (1/2, 1/2)}.

\end{proof}

\subsection{Asymptotic limits}

\noindent We now study the neutral drift and strong selection limits of the reward cost function $E_r(\theta)$ with $t=1$ for the all-defector initial state when the intensity of selection $\beta$ tends to $0$ and to +$\infty$, respectively.
\begin{proposition}(Neutral drift limit) It holds that
\label{prop: neutral drift t=1 (1, 0)}    
$$
\lim\limits_{\beta\rightarrow 0}E_r(\theta)=\frac{N^{2} \theta}{N-1} \Big(\frac{4P(N)}{4P(N)-Q(N)}\Big),
$$
where
\begin{align*}
P(N)&= 1+(-1)^1{N-2-1 \choose 1}\Big(\frac{1}{4}\Big)^1+(-1)^2{N-2-2 \choose 2}\Big(\frac{1}{4}\Big)^{2}+\ldots +(-1)^{\left\lfloor\frac{N-2}{2}\right\rfloor}{N-2-\lfloor\frac{N-2}{2}\rfloor \choose \lfloor\frac{N-2}{2}\rfloor}\Big(\frac{1}{4}\Big)^{\lfloor\frac{N-2}{2}\rfloor},\\
Q(N)& = 1+(-1)^{1}{N-3-1 \choose 1}\Big(\frac{1}{4}\Big)^1+(-1)^{2}{N-3-2 \choose 2}\Big(\frac{1}{4}\Big)^{2}+\ldots + (-1)^{\left\lfloor\frac{N-3}{2}\right\rfloor}{N-3-\lfloor\frac{N-3}{2}\rfloor \choose \lfloor\frac{N-3}{2}\rfloor}\Big(\frac{1}{4}\Big)^{\lfloor\frac{N-3}{2}\rfloor}.
\end{align*}
\end{proposition}
\begin{proof}
Recall that $E_r(\theta) = \frac{N^2\theta}{N-1} \Big(\frac{1}{1-ad\frac{y_3}{y_2}}\Big)$. Following similar lines as in the proof of Proposition \ref{prop: neutral drift t=1 (1/2, 1/2)}, we obtain 
\begin{comment}
Recalling again from Lemma \ref{lem: first formula of the cost function} that
\[
E_r(\theta)=\frac{N^2\theta}{N-}\Big(\frac{1}{1-ad\frac{y_3}{y_2}}\Big).
\]
According to the first statement of Lemma \ref{lem: recurrence relation}
\begin{align*}
 y_2&=\hat{y}_{N-2}= \sum\limits_{k=0}^{\lfloor\frac{N-2}{2}\rfloor}(-1)^{k}{N-2-k \choose k}(b d)^{k},\\
 y_3&=\hat{y}_{N-3}=\sum\limits_{k=0}^{\lfloor\frac{N-3}{2}\rfloor}(-1)^{k}{N-3-k \choose k}(b d)^{k}.
\end{align*}
Note that $a$, $b$, and $d$ depend on $\beta$, and therefore $y_2$ and $y_3$ also depend on $\beta$ through the product $(bd)$. We have
\begin{align*}
&\lim\limits_{\beta\rightarrow 0} (bd)= \lim\limits_{\beta\rightarrow 0}\frac{1}{2+e^{\beta \delta}+e^{-\beta \delta}}= \frac{1}{4},\\
&\lim\limits_{\beta\rightarrow 0} (ad)=\lim\limits_{\beta\rightarrow 0} \frac{1}{(1+e^{\beta \delta})(1+e^{-\beta(\theta+\delta)})}=\frac{1}{4}.
\end{align*}
It also follows that
\[
\lim\limits_{\beta\rightarrow 0} y_2= P(N), \quad \lim\limits_{\beta\rightarrow 0} y_3= Q(N),
\]
where $P(N)$ and $Q(N)$ are given explicitly in the statement of the Proposition.
\end{comment}
\begin{align*}
\lim _{\beta \rightarrow 0} E_{r}(\theta)&=\lim _{\beta \rightarrow 0} \Bigg(\frac{N^2\theta}{N-1}\Bigg(\frac{1}{1-ad\frac{y_3}{y_2}}\Bigg)\Bigg)
\\&=\frac{N^{2} \theta}{N-1} \left(\frac{1 }{1-\frac{1}{4}\frac{Q(N)}{P(N)}}\right)\\
&= \frac{N^{2} \theta}{N-1} \Big(\frac{4P(N)}{4P(N)-Q(N)}\Big).
\end{align*}

\end{proof}

\begin{figure}[H]
\centering
\hspace*{-2cm}  
\includegraphics[width=0.7\textwidth]{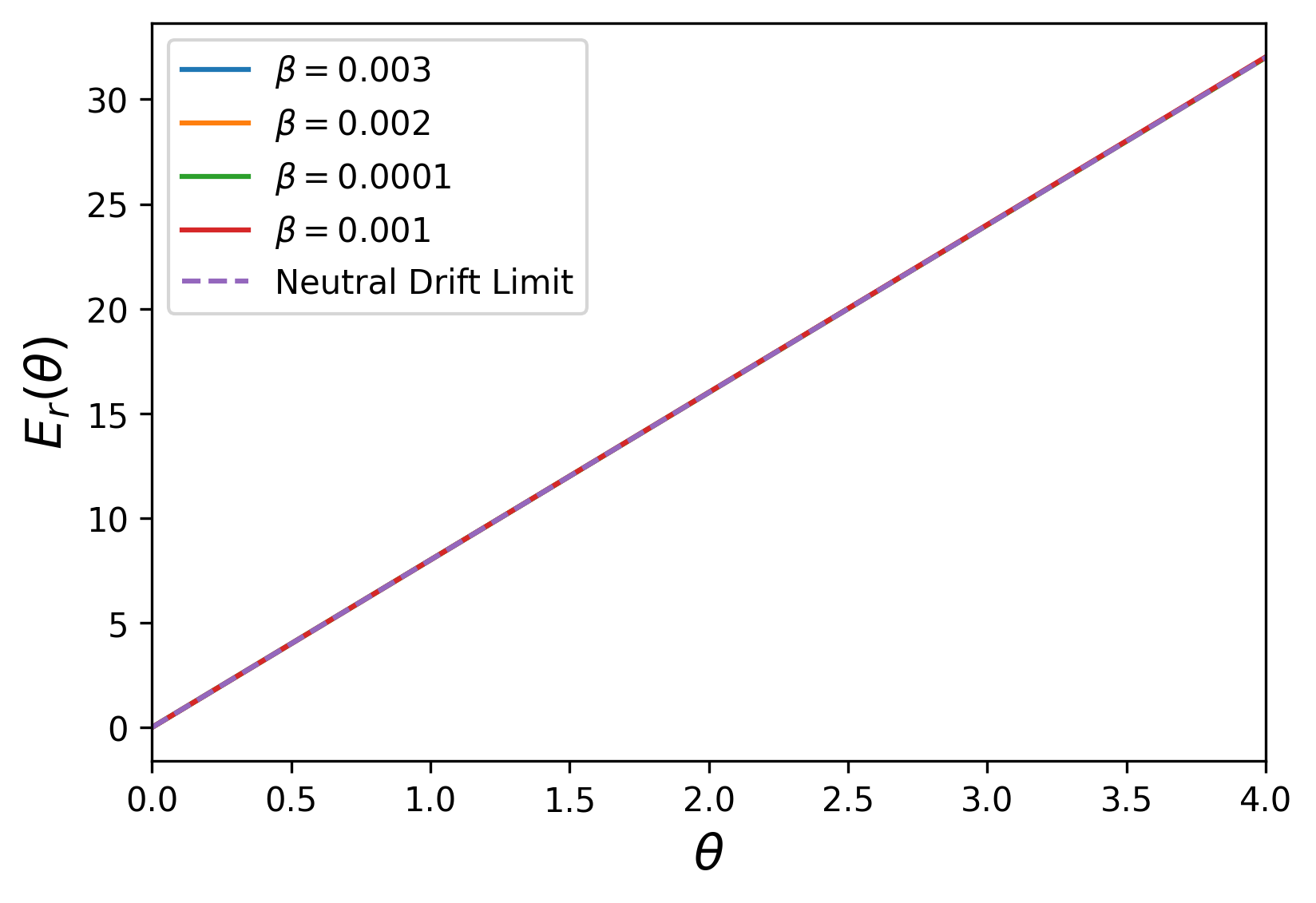}
\caption{The neutral drift limit (dashed purple line) for the reward cost function $E_{r}(\theta)$ with $N = 3$, $t = 1$, and $\theta = 1$ for DG with $B = 2, c = 1$, under the assumption that the population starts with all defectors. We notice that, the closer the $\beta$ value approaches 0, the closer to the limiting value the cost function $E_{r}(\theta)$ gets, in accordance to the analytical result Proposition \ref{prop: neutral drift t=1 (1, 0)}.} 
\label{fig:neutraldriftDG10}
\end{figure}

\begin{proposition}(Strong selection limit) It holds that
\label{prop: infinite selection t=1 (1, 0)}    
$$
\lim\limits_{\beta\rightarrow \infty}E_r(\theta)=\begin{cases}
+\infty \quad\text{if}\quad \delta+\theta> 0,\\
\frac{2 N^2\theta}{N-1}\quad\text{if}\quad \delta+\theta= 0,\\
\frac{N^2\theta}{N-1} \quad\text{if}\quad \delta+\theta< 0.
\end{cases}
$$
\end{proposition}
\begin{proof}
Recall that $E_r(\theta) = \frac{N^2\theta}{N-1} \Big(\frac{1}{1-ad\frac{y_3}{y_2}}\Big)$. Following similar lines as in the proof of Proposition \ref{prop: infinite selection t=1 (1/2, 1/2)}, we get 
\begin{comment}
We proceed as in the proof of the previous proposition by computing the limit of relevant quantities as $\beta\rightarrow \infty$ instead of $\beta\rightarrow 0$, noting that $\delta<0$ in our setting. We have
\begin{align*}
&\lim\limits_{\beta\rightarrow \infty} (bd)= \lim\limits_{\beta\rightarrow \infty}\frac{1}{2+e^{\beta \delta}+e^{-\beta \delta}}= 0,\\
&\lim\limits_{\beta\rightarrow \infty} (ad)=\lim\limits_{\beta\rightarrow \infty} \frac{1}{(1+e^{\beta \delta})(1+e^{-\beta(\theta+\delta)})}=\begin{cases}
1 \quad\text{if}\quad \delta+\theta> 0,\\
\frac{1}{2}\quad\text{if}\quad \delta+\theta= 0,\\
0 \quad\text{if}\quad \delta+\theta< 0.
\end{cases}
\end{align*}
It also follows that
\[
\lim\limits_{\beta\rightarrow \infty} y_2= 1, \quad \lim\limits_{\beta\rightarrow \infty} y_3=1.
\]
\end{comment}
\begin{align*}
\lim _{\beta \rightarrow \infty} E_{r}(\theta)&=\lim _{\beta \rightarrow \infty} \Bigg(\frac{N^2\theta}{N-1}\Big(\frac{1}{1-ad\frac{y_3}{y_2}}\Big)\Bigg)
\\&=\begin{cases}
+\infty \quad\text{if}\quad \delta+\theta> 0,\\
\frac{2 N^2\theta}{N-1}\quad\text{if}\quad \delta+\theta= 0,\\
\frac{N^2\theta}{N-1} \quad\text{if}\quad \delta+\theta< 0.
\end{cases}
\end{align*}
\end{proof}

\begin{figure}[H]
\centering
\hspace*{-2cm}  
\includegraphics[width=1.2\textwidth]{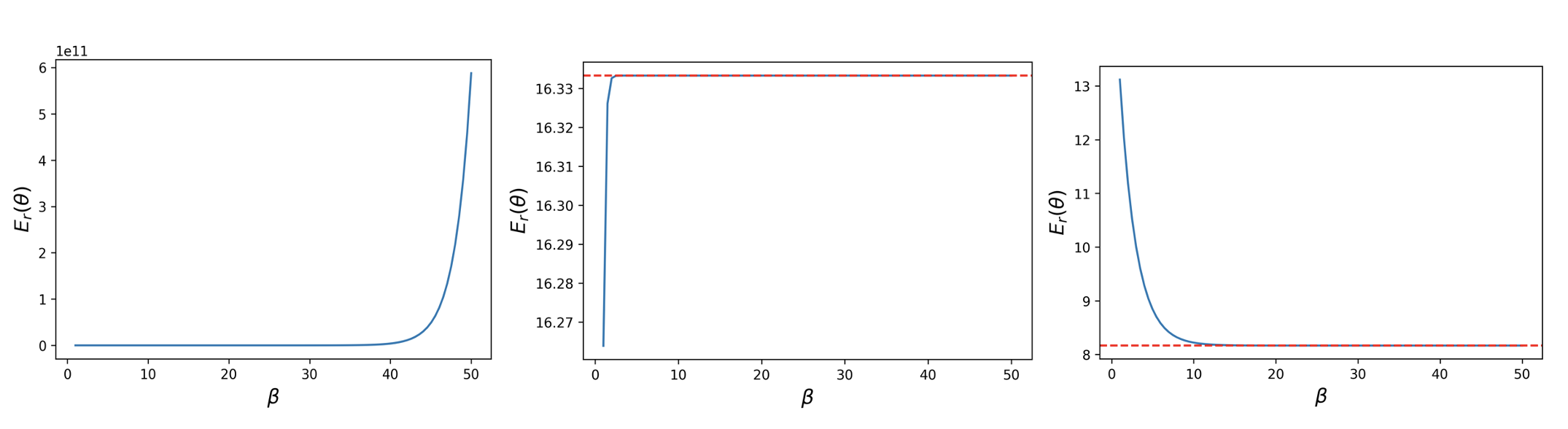}
\caption{The strong selection limit (dashed red line) for the reward cost function $E_{r}(\theta)$ with $N = 7$, $t = 1$ for DG with $B = 2, c = 1$, under the assumption that the population starts with all defectors. The first image corresponds to $\delta + \theta >0$ ($\delta = -0.5$, $\theta = 1$), the middle one to $\delta + \theta = 0$ ($\delta = - 1$, $\theta = 1$), and the last one to $\delta + \theta < 0$ ($\delta$ = -1.5, $\theta = 1$). The numerical results are in accordance with the analytical one in Proposition \ref{prop: infinite selection t=1 (1, 0)}.} 
\label{fig:strongselectionDG10}
\end{figure}
\color{black}

\subsection{Institutional reward with $t = 2$}
\label{sec: Reward t = 2}

We compute the reward cost function for $t = 2$, i.e., when the institution provide rewards only if there are at most \textit{two} cooperators in the population, under both starting states.

\noindent For $t=2$, \\

$W=\begin{pmatrix}
1&-a&&&&&\\
-c&1&-a&&&&&\\
&-d&1&-b&&&\\
&&\ddots&\ddots&\ddots&\\
&&&-d&1&-b&&&\\
&&&&-d&1&-b&&&\\
&&&&&\ddots&\ddots&\ddots&\\
&&&&&&-d&1&-b\\
&&&&&&&-d&1
\end{pmatrix}.$
\\

Therefore, for the equally likely starting state, by using Equation \eqref{eq: total investment reward (1/2,1/2)}, we get
\begin{align}
\label{eq: Reward t = 2 (1/2,1/2)}
E_r(\theta)&=\frac{N^{2} \theta}{2} \sum_{j=1}^{2} \frac{\left(W^{-1}\right)_{1, j}+\left(W^{-1}\right)_{N-1, j}}{N-j} \nonumber
\\&=\frac{N^{2} \theta}{2}\left[\left(\frac{\left(W^{-1}\right)_ {1,1}+\left(W^{-1}\right)_{N-1,1}}{N-1}\right) +\left(\frac{\left(W^{-1}\right)_{1,2} +\left(W^{-1}\right)_{N-1,2}}{N-2}\right)\right]. 
\end{align}
To obtain a simplified version of Equation \eqref{eq: Reward t = 2 (1/2,1/2)}, we need to compute $(W^{-1})_{1,1}$, $(W^{-1})_{N-1,1}$, $(W^{-1})_{1,2}$, $(W^{-1})_{N-1,2}$. We apply Theorem \ref{thm: diag element}, the diagonal element case, and Corollary \ref{cor: fi product}, both cases, to get: 
\begin{align*}
(W^{-1}_{1,1})&= \frac{1}{1-ac\frac{y_3}{y_2}},\\
(W^{-1})_{N-1,1}&= c d^{N-3} \frac{1}{y_{2}} \frac{1}{1-a c \frac{y_{3}}{y_{2}}},\\
(W^{-1})_{1,2}&= a \frac{1}{1-a c-a d\frac{y_4}{y_3}},\\
(W^{-1})_{N-1,2}&= d^{N-3}\frac{1}{y_{3}} \frac{1}{1-a c-a d \frac{y_4}{y_3}}.
\end{align*}

See Section \ref{sec: Reward t = 2 Ws} for detailed computations of the $W$ entries.

Substituting the $W$ values in Equation \eqref{eq: Reward t = 2 (1/2,1/2)} yields:

\begin{align*}
E_r(\theta)&=\frac{N^{2} \theta}{2}\left(\frac{\frac{1}{1-a c \frac{y_{3}}{y_{2}}}+c d^{N-3} \frac{1}{y_{2}} \frac{1}{1-a c \frac{y_3}{y_2}}}{N-1}+\frac{a\frac{1}{1-ac-ad\frac{y_4}{y_3}}+d^{N-3}\frac{1}{y_3}\frac{1}{1-ac-ad\frac{y_4}{y_3}}}{N-2}\right)
\\& =\frac{N^{2} \theta}{2}\left[\frac{\frac{1}{1-a c \frac{y_{3}}{y_{2}}}\left(1+c d^{N-3} \cdot \frac{1}{y_{2}}\right)}{N-1} + \frac{\frac{1}{1-a c-a d \frac{y_{4}}{y_{3}}}\left(a+d^{N-3} \cdot \frac{1}{y_{3}}\right)}{N-2}\right].
\end{align*}

\noindent For the all-defector starting state, we have:
\begin{align}
\label{eq: Reward t = 2 (1,0)}
E_r(\theta)&=\sum_{j=1}^{2} \frac{\left(W^{-1}\right)_{1, j}}{N-j} =\frac{W^{-1}_{1,1}}{N-1}+\frac{W^{-1}_{1,2}}{N-2}. 
\end{align}

\noindent By substituting the $W$ values in Equation \eqref{eq: Reward t = 2 (1,0)}, we get:

$$
E_r(\theta)=\frac{\frac{1}{1-a c \frac{y_{3}}{y_{2}}}}{N-1}+\frac{a\frac{1}{1-ac-ad\frac{y_4}{y_3}}}{N-2}.
$$

\section{Institutional incentives under general starting state and mutation}
\label{sec: general mutation}

In this section, we present the cost function $E(\theta)$ for the case of  general mutation rates. Indeed, for an arbitrary mutation rate $\mu$, the transition probabilities change as follows.
\begin{equation*} 
\label{eq: transition probabilities reward general mutation}
\begin{split} 
u_{i,i\pm k} &= 0 \qquad \text{ for all } k \geq 2, \\
u_{i,i+1} &= \frac{N-i}{N}\left(\mu + (1-\mu)\frac{i}{N} \left(1 + e^{-\beta[\Pi_C(i) - \Pi_D(i)+\theta_i/i]}\right)^{-1}\right),\\
u_{i,i-1} &= \frac{i}{N}\left(\mu + (1-\mu)\frac{N-i}{N} \left(1 + e^{\beta[\Pi_C(i) - \Pi_D(i)+\theta_i/i]}\right)^{-1}\right),\\
u_{i,i} &= 1 - u_{i,i+1} -u_{i,i-1}.
\end{split} 
\end{equation*}
Thus, the expected cost function $E(\theta)$ given the initial state of the population being $S_i$ (i.e. with $i$ cooperators and $N-i$ defectors) is
\begin{align}
\label{eq: total investment reward general mutation}
E^i(\theta) = \sum_{j=0}^{N} n_{ij}  \theta_j, 
\end{align}
where $n_{i;j}$ is obtained from $\mathcal{N}=(n_{ij})_{i,j=0}^{N}= (I-U)^{-1}$ and
\begin{equation}
\label{eq: incentives per generation general mutation}
\theta_j = \begin{cases} \frac{j}{a}\theta,\quad \text{reward incentive},\\
\frac{N-j}{b}\theta,\quad \text{punishment incentive},\\
\min\Big(\frac{j}{a}, \frac{N-j}{b}\Big)\theta,\quad\text{mixed incentive}.
\end{cases}
\end{equation}
\noindent Depending on the incentive $\theta_j$ selected, $E^i(\theta)$ could become  either $E^i_r(\theta)$, $E^i_p(\theta)$, or $E^i_{mix}(\theta)$. In Section \ref{sec: behaviour general mutation}, we perform numerical simulations on the reward cost function $E^i_r(\theta)$ for various starting points.
\color{black}

\section{Numerical Analysis}
\label{sec: numerical analysis}

In this section, we present the results of our numerical analysis, coded in Python 3.10. We employ a logarithmic scale to respond to skewness of larger values.

\subsection{Behaviour of the cost functions under small mutation with different initial states}

In Figures \ref{fig:varioustDGold} and \ref{fig:varioustDGnew}, we plot the reward cost function for DG for the two initial state assumptions: the dynamics randomly commencing in the state $S_0$ or the state $S_N$ and it starting in the state $S_0$.  Figures \ref{fig:varioustPGGold} and 
\ref{fig:varioustPGGnew} in  Appendix, show the behaviour of the reward cost function for PGG, while Figures \ref{fig:varioustDGpunishmentold} and 
\ref{fig:varioustDGpunishmentnew} show that of the punishment cost function for DG. In Figures \ref{fig:varioustDGrewardandpunishmentold} and  \ref{fig:varioustDGrewardandpunishmentnew}, we plot the behaviour of the hybrid cost function for DG. The numerical simulations concern some small values of $N$ and various values of the other parameters. From these plots, we observe that the cost functions $E_r(\theta)$, $E_p(\theta)$ and $E_{mix}(\theta)$ are increasing for all $1 \leq t < N-1$ and become non-monotonic for $t = N-1$, exhibiting a phase transition. This behaviour is dependant on changes in the strength of selection, $\beta$. The larger this parameter is, the more pronounced the non-monotonic behaviour of the function becomes. The aforementioned behaviour is robust to changes in the game-specific values such as $B$, $c$, $r$, $n$.

%\textcolor{red}{Note that, in the case of DG with $N = 3$, $t = 1$ and $N = 4$, $t = 1, 2$ as well as in the case of PGG with $N = 6$ and $t = 1,2,3,4$, we applied a logarithmic scale on the $y$-axis.}

%\textcolor{red}{Here, we looked both at the equally likely starting point and the starting with all defectors one. Moreover, we performed numerical analysis on the general mutation reward cost function \ref{eq: total investment reward general mutation}.}
\begin{figure}[H]
\centering
\hspace*{-1cm}  
\includegraphics[width=0.8\textwidth]{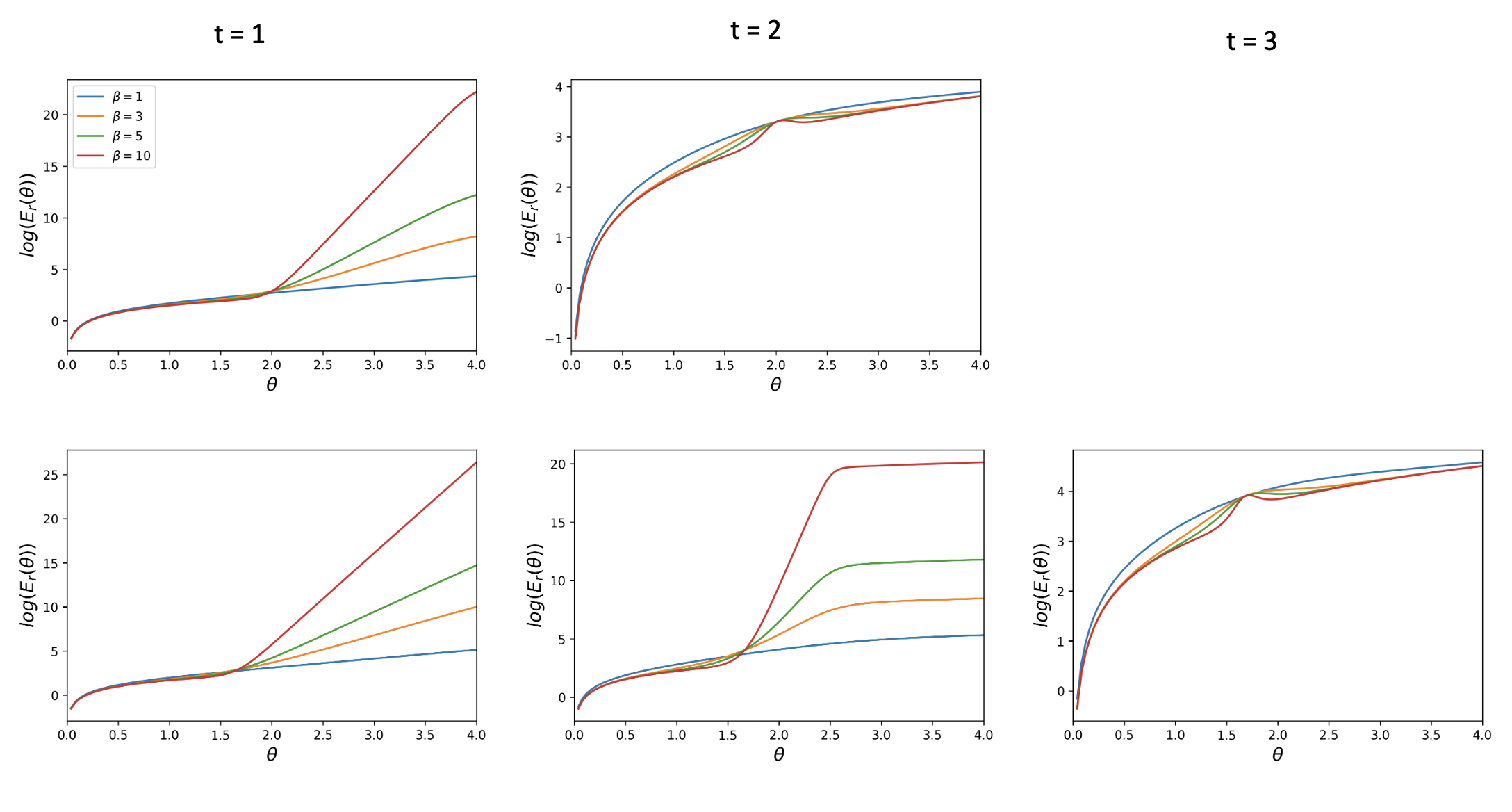}
\caption{Behaviour of the reward cost function $E_{r}(\theta)$, for different thresholds $t$ and strengths of selection $\beta$, for DG with $B = 2, c = 1$. The first row corresponds to $N = 3$ and the second one to $N = 4$. The leftmost column corresponds to $t = 1$, the middle one to $t = 2 $, the rightmost one to $t = 3$. This is for the assumption that the
population is equally likely to start in the homogeneous state $S_0$ as well as in the homogeneous
state $S_N$.}
\label{fig:varioustDGold}
\end{figure}

\begin{figure}[H]
\centering
\hspace*{-1cm}  
\includegraphics[width=0.8\textwidth]{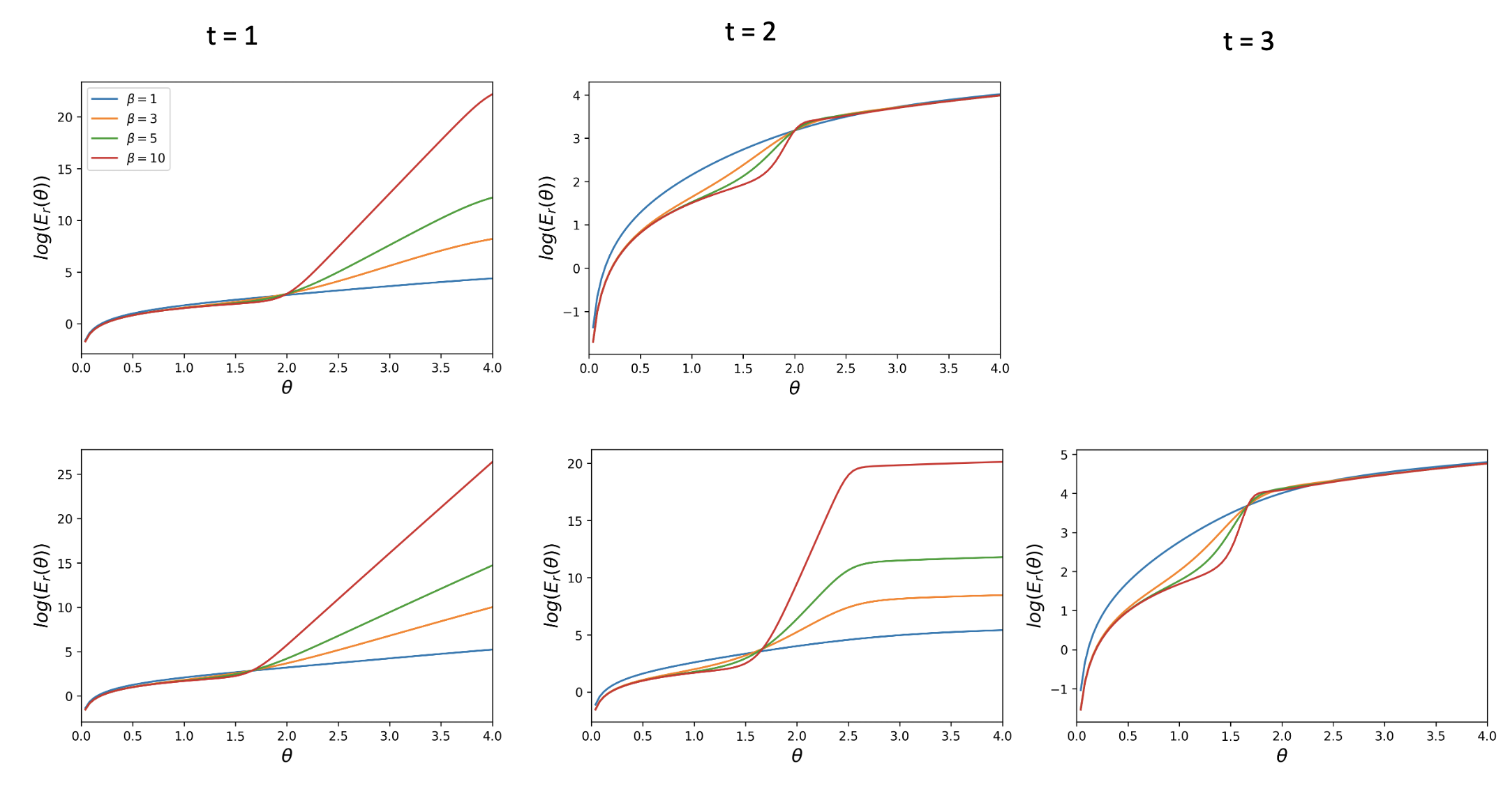}
\caption{Behaviour of the reward cost function $E_{r}(\theta)$, for different thresholds $t$ and strengths of selection $\beta$, for DG with $B = 2, c = 1$. The first row corresponds to $N = 3$ and the second one to $N = 4$. The leftmost column corresponds to $t = 1$, the middle one to $t = 2 $, the rightmost one to $t = 3$. This is for the assumption that the population is expected to start with all defectors, i.e. in state $S_0$.}
\label{fig:varioustDGnew}
\end{figure}

\subsection{Behaviour of the reward cost function under general mutation with various starting points}
\label{sec: behaviour general mutation}

In Figures \ref{fig:varioustDGrewardgeneralmutation0}, \ref{fig:varioustDGrewardgeneralmutationN/2} and  \ref{fig:varioustDGrewardgeneralmutationN-1}, we plot the reward cost function (as a function of the per capital cost $\theta)$  for the evolutionary processes with general mutation and different starting states $S_0, S_{\frac{N}{2}}, S_{N-1}$ for $t=1$, $t=N-1$ (for some $N$), and various values of other parameters. These figures clearly show the strong and non-trivial impact of the threshold $t$, the mutation rate $\mu$, and the strength of selection $\beta$ on the cost function. We observe that for $t=1$, the cost function is decreasing with respect to $\mu$. For $t=N-1$, the function is also decreasing with respect to $\mu$ for sufficiently small $\beta$, but for sufficiently large $\beta$ it behaves much more complicatedly.
To be comparable to the paper's main analysis using the small mutation approach, we assume that the mutation rate is small enough that the cost is calculated until a homogeneous state is reached. This is in line with most previous simulation works on incentive optimisation (see again the Remark).

%The non-monotonic behaviour of the cost function $E_r(\theta)$ observed in the case of small mutation is preserved in the case of arbitrary mutation rates as follows: the function is increasing for $1 \leq t < N-1$ and for small values of $\beta$, then becomes non-monotonic for $t = N-1$ and for larger values of $\beta$, exhibiting a phase transition.\\
%{Moreover, we note that the threshold $t$ and the mutation rate $\mu$ impact the amount the external decision-maker needs to invest in order to achieve a desired level of cooperation. A lower threshold value and a larger mutation rate incur a larger cost, as seen in Figure 7 row 1 .... Similarly a higher threshold value and a smaller mutation rate incurs a larger cost, as observed in Figure 7 row 2...}

\begin{figure}[H]
\centering
\hspace*{-1cm}  
\includegraphics[width=0.8\textwidth]{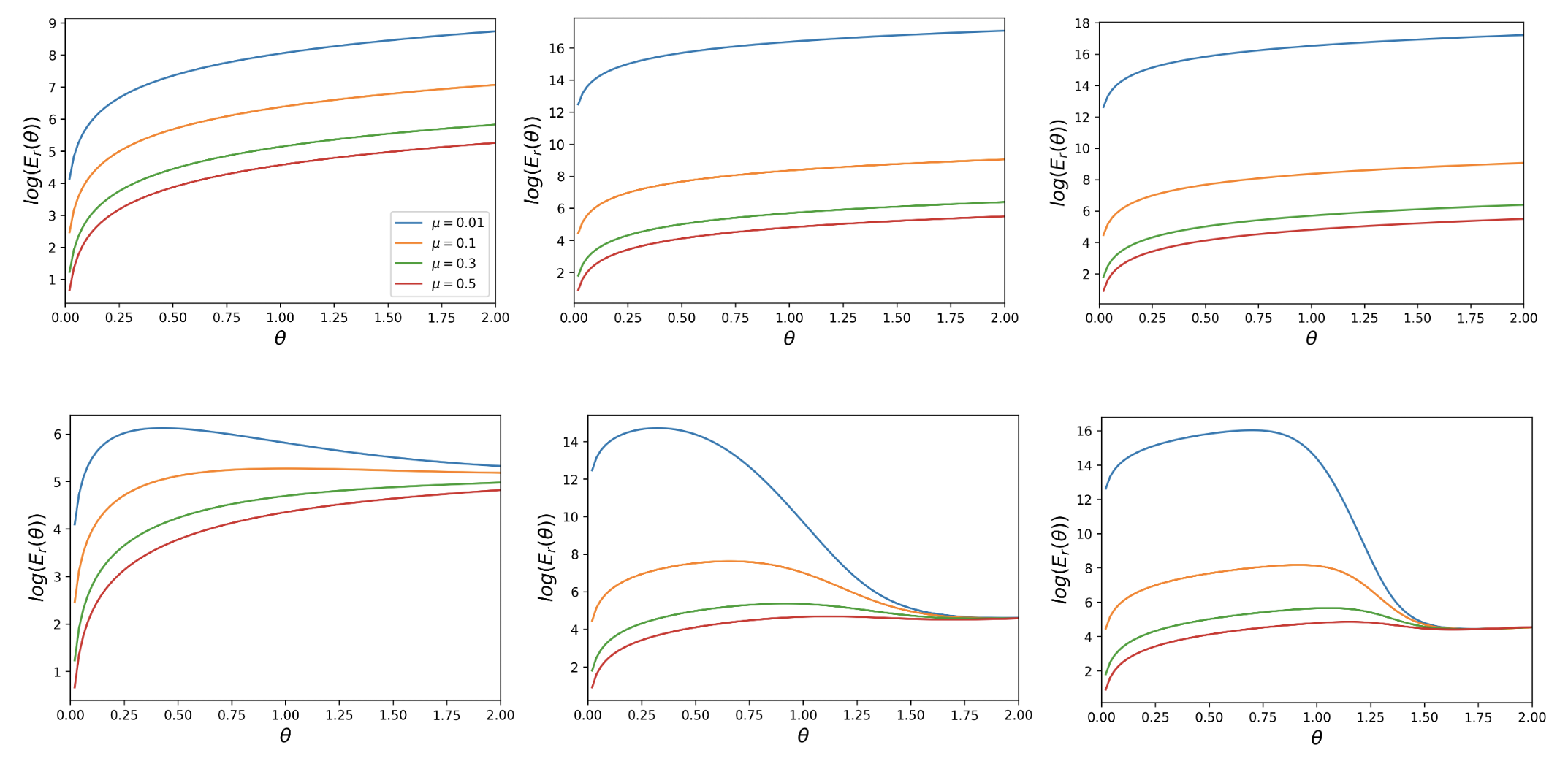}
\caption{Behaviour of the reward cost function $E_r(\theta)$, for $N = 6$, $t = 1, 5$, and strengths of selection $\beta = 1, 5, 10$, for DG with $B = 2, c = 1$. The first row corresponds to $t = 1$, while the second one to $t = 5$ $(N-1)$. The first column corresponds to $\beta = 1$, the middle one to $\beta = 5$, and the last one to $\beta = 10$. The behaviour of the reward cost function for PGG is similar. This is for a general mutation rate $\mu$. The starting point is $S_0$.}
\label{fig:varioustDGrewardgeneralmutation0}
\end{figure}

\begin{figure}[H]
\centering 
\hspace*{-1cm}  
\includegraphics[width=0.8\textwidth]{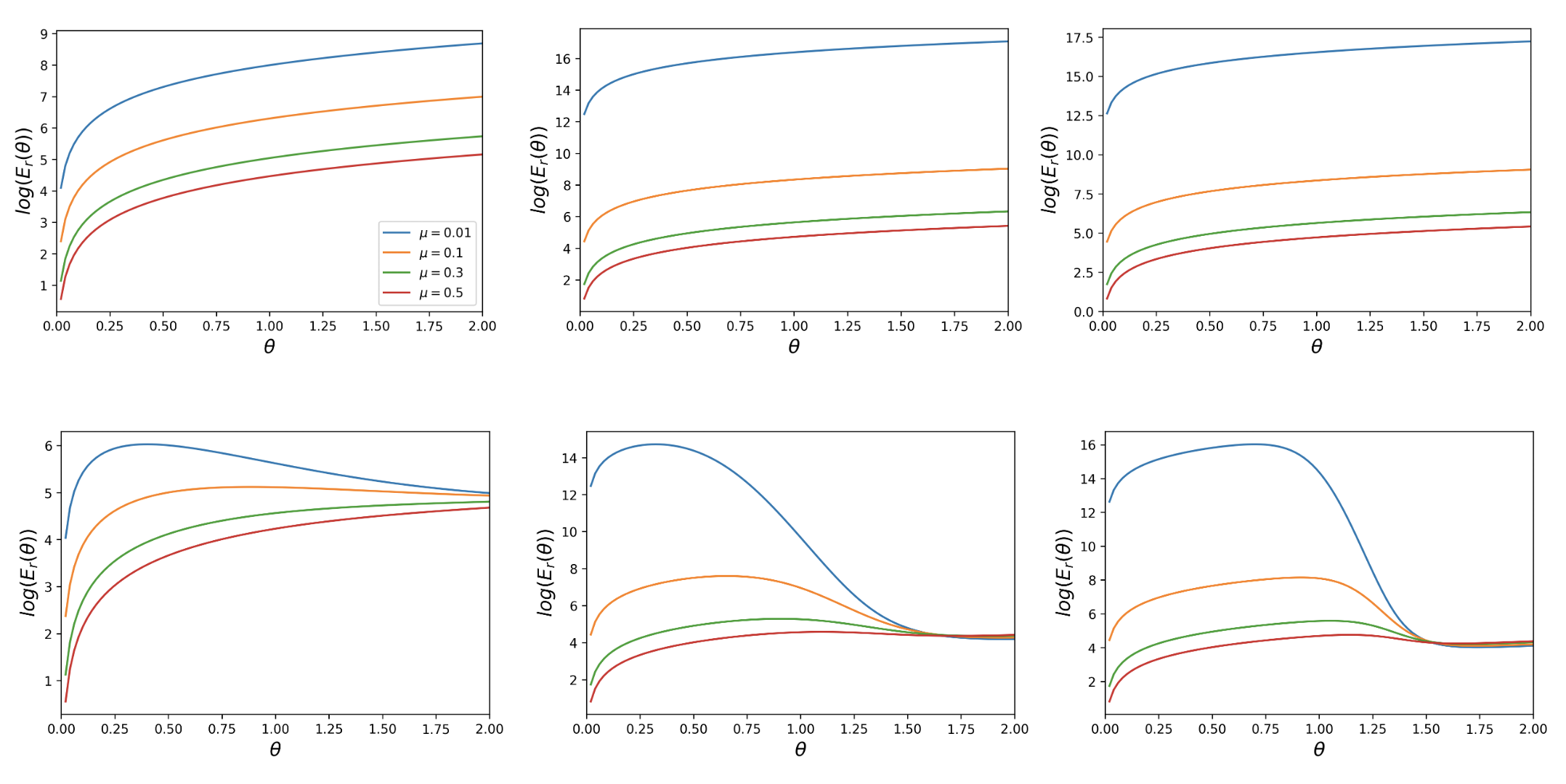}
\caption{Behaviour of the reward cost function $E_r(\theta)$, for $N = 6$, $t = 1, 5$, and strengths of selection $\beta = 1, 5, 10$, for DG with $B = 2, c = 1$. The first row corresponds to $t = 1$, while the second one to $t = 5$ $(N-1)$. The first column corresponds to $\beta = 1$, the middle one to $\beta = 5$, and the last one to $\beta = 10$. The behaviour of the reward cost function for PGG is similar. This is for a general mutation rate $\mu$. The starting point is $S_{\frac{N}{2}}$.}
\label{fig:varioustDGrewardgeneralmutationN/2}
\end{figure}

\begin{figure}[H]
\centering
\hspace*{-1cm}  
\includegraphics[width=0.8\textwidth]{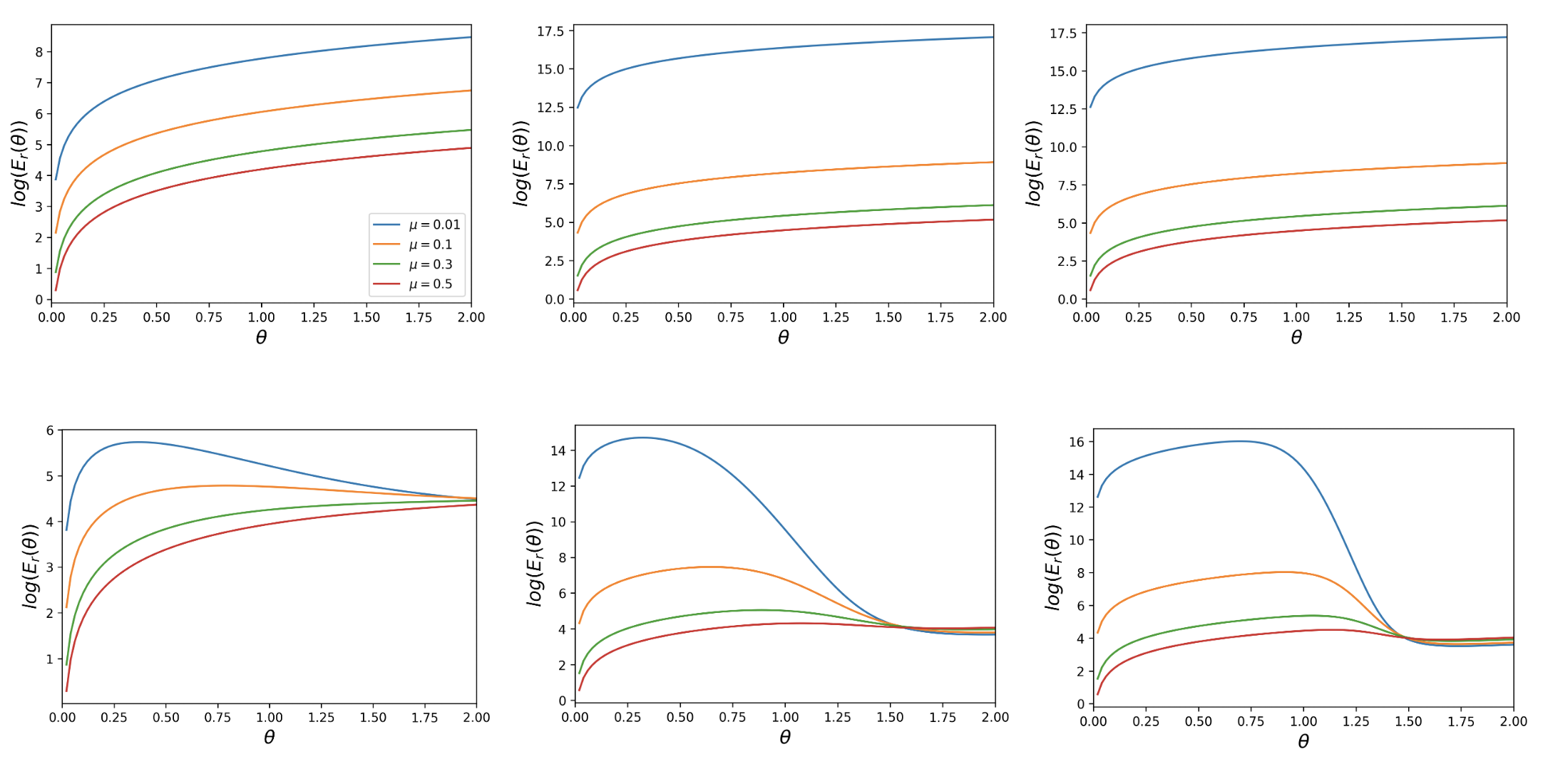}
\caption{Behaviour of the reward cost function $E_r(\theta)$, for $N = 6$, $t = 1, 5$, and strengths of selection $\beta = 1, 5, 10$, for DG with $B = 2, c = 1$. The first row corresponds to $t = 1$, while the second one to $t = 19$ $(N-1)$. The first column corresponds to $\beta = 1$, the middle one to $\beta = 5$, and the last one to $\beta = 10$. The behaviour of the reward cost function for PGG is similar. This is for a general mutation rate $\mu$. The starting point is $S_{N-2}$.}
\label{fig:varioustDGrewardgeneralmutationN-1}
\end{figure}

\subsection{Phase transition: change in the qualitative behaviour of the cost function}
Theorems \ref{thm: derivative positive t1 (1/2, 1/2)} and \ref{thm: derivative positive t1 (1, 0)} show that the reward cost function for $t=1$ is always non-decreasing for all values of the intensity of selection $\beta$. In \cite{duong2021cost}, it is shown that, for $t=N-1$, $E_r(\theta)$ and $E_p(\theta)$ are non-decreasing when $\beta$ is sufficiently small, but are not monotonic when $\beta$ is large enough. Furthermore,  in \cite{DuongDurbacHan2022}, the same behaviour is proven true for $E{mix}(\theta)$ when $t=N-1$. This demonstrates that the qualitative behaviour of the cost function changes significantly when $t$ and $\beta$ vary. We conjecture that there exists a critical threshold value of $t^*$ such that:
for $t\leq t^*$, $E(\theta)$ (where $E(\theta)$ can be either $E_r(\theta)$, $E_p(\theta)$, or $E_{mix}(\theta)$) is always non-deceasing for all $\beta$ when $t\leq t^*$, while for $t^*<t\leq N-1$, $E(\theta)$ is non-decreasing when $\beta$ is sufficiently small, but is not monotonic when $\beta$ is sufficiently large. Figures \ref{fig:varioustDGnew} - \ref{fig:varioustDGrewardandpunishmentold} suggest that, for small population size $N$, the critical threshold $t^*$ is $N-1$. How to prove this interesting phase transition phenomena for general $N$ is elusive to us at the moment and deserves further investigation in the future.  

All the simulations in Section \ref{sec: numerical analysis} can be found in the  \changeurlcolor{blue}\href{https://github.com/calinadurbac/Evolutionary-Game-Theory/tree/main/Reward%20and%20Punishment%20-%20General%20t}{`Evolutionary Game Theory' repository under the `Reward and Punishment - General t' folder}.

%We present a phase diagram showing the change in monotonicity in relation to $t$ as well as numerical calculation of the value of $t$ at the point of change. 
\begin{comment}
\begin{figure}[H]
\centering
\includegraphics[width=1\textwidth]{phase diagram.png}
\caption{We show the behaviour of the reward cost function, $E_{r}(\theta)$}, in connection to the value of the threshold $t$. It clearly exhibits a phase transition. 
\label{fig:phase diagram}
\end{figure}
\end{comment}

\section{Discussion}
\label{sec: discussion}

Over the past decades, there has been a lot of attention given to studying effective incentive mechanisms with the aim of promoting cooperation. Various mechanisms have been tested \cite{nowak2006,sigmund2010calculus,perc2017statistical,rand2013human,van2014reward}, with some of the most efficient ones being institutional incentives, where there is a central decision-maker in charge of applying them. In our model, we adopt this idea of an external decision-maker and entrust them with providing reward, punishment, or hybrid incentives to players interacting via two cooperation dilemmas, the Donation Game or the Public Goods Game. 
While other works have examined a comparable setting, relatively few have looked at the question of optimising the overall cost to the institution while maintaining a certain degree of cooperation. Moreover, most studies have focused on  the full-invest approach (also known as the standard institutional incentive model), in which incentives are always provided regardless of the population composition.

In this paper, we studied the problem of optimising the cost of institutional incentives that are provided conditionally on the number of cooperators in the population (namely, when less than a threshold $t$) while guaranteeing a certain level of cooperation in a well-mixed, finite population of selfish individuals. We use mathematical analysis to derive the cost function as well as the neutral drift and strong selection limits for the case $t = 1$ and the cost function for the case $t = 2$ (using two different initial states). We provide numerical investigation for the aforementioned cases as well as for the others, i.e., for $2 \leq t < N-1$. We also introduced cost functions with a general mutation rate for reward, punishment, and hybrid incentives and numerically analysed their behaviour.

For the mathematical analysis of the reward incentive cost function to be possible, we made some assumptions. Firstly, in order to derive the analytical formula for the frequency of cooperation, we assumed a small mutation limit \cite{rockenbach,nowak2004emergence,sigmund2010calculus}. Despite the simplified assumption, this small mutation limit approach has  wide applicability to scenarios which go well beyond the strict limit of very small mutation rates \cite{zisis2015generosity,hauert2007,sigmundinstitutions,rand2013evolution, DuongDurbacHan2022}. If we were to relax this assumption, the derivation of a closed form for the frequency of cooperation would be intractable. Secondly, we focused on two important cooperation dilemmas, the Donation Game and the Public Goods Game. Both have in common that the difference in average payoffs between a cooperator and a defector does not depend on the population composition. This special property allowed us to simplify the fundamental matrix of the Markov chain to a tri-diagonal form and apply the techniques of matrix analysis \cite{huang1997} to obtain a closed form of its inverse matrix. In games with more complex payoff matrices such as the general prisoner's dilemma and the collective risk game \cite{sun2021combination}, this property no longer holds (e.g., in the former case, the payoff difference, $\Pi_C(i)-\Pi_D(i)$, depends additively on $i$) and the technique in this paper cannot be directly applied. In these scenarios, we might consider other approaches to approximate the inverse matrix, exploiting its block structure.

We intend to utilise analytical techniques to explore the optimisation problems of punishment and hybrid incentives (reward and punishment used concurrently) for the individual-based incentive scheme. This would be interesting because it would allow for a cost comparison between reward or punishment incentives for a certain threshold value $t$ and the mixed scheme for the same value of $t$. Finally, there has been little attention given to the use of analysis for obtaining insights into cost-efficient incentives in structured populations or in more complex games (such as the general prisoner's dilemma and the collective risk game), so this would also be an engaging research avenue.

\newpage
 \renewcommand{\thefigure}{A\arabic{figure}}
 \renewcommand{\thetable}{A\arabic{table}}
 \setcounter{figure}{0}  

\section{Appendix}
\label{sec: appendix}

In this appendix, we provide explicit computations of the reward cost function $E_r(\theta)$ for small populations. We also include the calculations of the $W$ entries for reward $t = 2$ and some numerical simulations.

\subsection{Small population examples for reward}

\subsubsection{$N=3, t=1$}

For $N=3, t=1$, we have

$$
W = \begin{pmatrix}
1&-a\\
-d&1\\
\end{pmatrix}
$$

and 

$$
W^{-1} = \begin{pmatrix}
\frac{-1}{ad-1}&\frac{-a}{ad-1}\\
\frac{-d}{ad-1}&\frac{-1}{ad-1}\\
\end{pmatrix}.
$$

Thus, the cost function is
\begin{align*}
E_r(\theta) &= \frac{N^2\theta}{2} \sum_{j=1}^{1}\frac{(W^{-1})_{1j} + (W^{-1})_{N-1,j}}{N-j} 
\\&= \frac{9\theta}{4}(W^{-1}_{11}+W^{-1}_{21})
\\&= \frac{9\theta}{4}\frac{-d-1}{ad-1}
\\&= \frac{9\theta}{4}\Bigg(\frac{1}{1-\frac{1}{(1+e^{-x})(e^y+1)}}+\frac{1}{\Big(1-\frac{1}{(1+e^{-x})(e^y+1)}\Big)(e^y+1)}\Bigg),
\end{align*} where $x=\beta(\delta+\theta)$ and $y=\beta\delta$.
%We compute the derivative of $E$ with respect to $\theta$
%$$
%E'(\theta)=\frac{1}{2}\Big(\frac{d+1}{-ad+1}+\theta\frac{d(d+1)}{(1-ad)^2}\frac{\beta e^{-\beta(\delta+\theta)}}{(1+e^{-\beta(\delta+\theta)})^2)}\Big)>0,
%$$
%since $ad<1$.
\begin{comment}
Its derivative is 
\begin{align*}
E'(\theta) &= \frac{9 \theta \left(\frac{\beta e^{- \beta \left(\delta + \theta\right)}}{\left(1 - \frac{1}{\left(1 + e^{- \beta \left(\delta + \theta\right)}\right) \left(e^{\beta \delta} + 1\right)}\right)^{2} \left(1 + e^{- \beta \left(\delta + \theta\right)}\right)^{2} \left(e^{\beta \delta} + 1\right)} + \frac{\beta e^{- \beta \left(\delta + \theta\right)}}{\left(1 - \frac{1}{\left(1 + e^{- \beta \left(\delta + \theta\right)}\right) \left(e^{\beta \delta} + 1\right)}\right)^{2} \left(1 + e^{- \beta \left(\delta + \theta\right)}\right)^{2} \left(e^{\beta \delta} + 1\right)^{2}}\right)}{4} \\& \qquad  + \frac{9}{4 \cdot \left(1 - \frac{1}{\left(1 + e^{- \beta \left(\delta + \theta\right)}\right) \left(e^{\beta \delta} + 1\right)}\right)} \\&\qquad
+ \frac{9}{4 \cdot \left(1 - \frac{1}{\left(1 + e^{- \beta \left(\delta + \theta\right)}\right) \left(e^{\beta \delta} + 1\right)}\right) \left(e^{\beta \delta} + 1\right)}.   
\end{align*}
       
\end{comment}
\subsubsection{$N=3, t=2$}
For $N = 3$, $t=2$, we have

$$
W = \begin{pmatrix}
1&-a\\
-c&1\\
\end{pmatrix}
$$

and 

$$
W^{-1} = \begin{pmatrix}
\frac{-1}{ac-1}&\frac{-a}{ac-1}\\
\frac{-c}{ac-1}&\frac{-1}{ac-1}\\
\end{pmatrix}.
$$

Hence, the cost function is 
\begin{align*}
E_r(\theta) &= \frac{N^2\theta}{2} \sum_{j=1}^{2}\frac{(W^{-1})_{1j} + (W^{-1})_{N-1,j}}{N-j} 
\\&= \frac{N^2\theta}{2}\big[\frac{(W^{-1}_{11}+W^{-1}_{21})}{2}+(W^{-1}_{12}+W^{-1}_{22})\big]
\\&= \frac{9\theta}{4}\Big(\frac{-c-1}{ac-1} \Big) +  \frac{9\theta}{2}\Big(\frac{-a-1}{ac-1}\Big)
\\&=  \frac{9\theta}{4}\Bigg(\frac{1}{1-\frac{1}{(1+e^{-x})(1+e^x)}} + \frac{1}{\Big(1-\frac{1}{(1+e^{-x})(1+e^x)}\Big)(1+e^{x})}\Bigg) \\&\qquad + \frac{9\theta}{2}\Bigg(\frac{1}{1-\frac{1}{(1+e^{-x})(1+e^x)}} + \frac{1}{\Big(1-\frac{1}{(1+e^{-x})(1+e^x)}\Big)(1+e^{-x})}\Bigg),
%\\&= \frac{\theta}{2}\Big[\frac{e^{y-x}+2e^{-x}+e^{y}+2}{e^{y-x}+e^{-x}+e^{y}} + 2\frac{2e^x+e^{-x}+3}{e^x+e{-x}+1}\Big]
%\\& = \frac{\theta}{2}\frac{5e^{x+y}+10e^{y-x}+3e^{y-2x}+2e^x+10e^{-x}+4e^{-2x}+2e^{y}+8}{e^{x+y}+2e^{y-x}+e^{y-2x}+e^{-x}+e^{-2x}+2e^{y}+1},
\end{align*} where $x=\beta(\delta+\theta)$ and $y=\beta\delta$.

\subsubsection{$N=4$, $t=1$}

For $N=4$, $t=1$, we have 

$$
W = \begin{pmatrix}
1&-a&0\\
-d&1&-b\\
0&-d&1\\
\end{pmatrix}
$$

and 

$$
W^{-1} = \begin{pmatrix}
\frac{bd-1}{ad+bd-1}&\frac{-a}{ad+bd-1}&\frac{-ab}{ad+bd-1}\\
\frac{-d}{ad+bd-1}&\frac{-1}{ad+bd-1}&-\frac{-b}{ad+bd-1}\\
\frac{-d^2}{ad+bd-1}&\frac{-d}{ad+bd-1}&\frac{ad-1}{ad+bd-1}\\
\end{pmatrix}. 
$$

Thus, the cost function is 
\begin{align*}
E_r(\theta) &= \frac{N^2\theta}{2} \sum_{j=1}^{1}\frac{(W^{-1})_{1j} + (W^{-1})_{N-1,j}}{N-j} 
\\&= \frac{8\theta}{3}(W^{-1}_{11}+W^{-1}_{31})
\\&= \frac{8\theta}{3}\frac{-d^2+bd-1}{ad+bd-1}
\\&= \frac{8\theta}{3}\Bigg(\frac{-1+\frac{1}{(1+e^{-y})(e^y+1)}}{-1+\frac{1}{(1+e^{-y})(e^y+1)}+\frac{1}{(1+e^{-x})(e^y+1)}} - \frac{1}{(e^y+1)^2\Big(-1+\frac{1}{(1+e^{-y})(e^y+1)}+\frac{1}{(1+e^{-x})(e^y+1)}\Big)} \Bigg),
\end{align*} where $x=\beta(\delta+\theta)$ and $y=\beta\delta$.

\begin{comment}
Its derivative is 

\begin{align*}
E'(\theta) &= \frac{8 \theta \left(- \frac{\beta \left(-1 + \frac{1}{\left(1 + e^{- \beta \delta}\right) \left(e^{\beta \delta} + 1\right)}\right) e^{- \beta \left(\delta + \theta\right)}}{\left(1 + e^{- \beta \left(\delta + \theta\right)}\right)^{2} \left(e^{\beta \delta} + 1\right) \left(-1 + \frac{1}{\left(1 + e^{- \beta \left(\delta + \theta\right)}\right) \left(e^{\beta \delta} + 1\right)} + \frac{1}{\left(1 + e^{- \beta \delta}\right) \left(e^{\beta \delta} + 1\right)}\right)^{2}} + \frac{\beta e^{- \beta \left(\delta + \theta\right)}}{\left(1 + e^{- \beta \left(\delta + \theta\right)}\right)^{2} \left(e^{\beta \delta} + 1\right)^{3} \left(-1 + \frac{1}{\left(1 + e^{- \beta \left(\delta + \theta\right)}\right) \left(e^{\beta \delta} + 1\right)} + \frac{1}{\left(1 + e^{- \beta \delta}\right) \left(e^{\beta \delta} + 1\right)}\right)^{2}}\right)}{3} \\&\qquad + \frac{8 \left(-1 + \frac{1}{\left(1 + e^{- \beta \delta}\right) \left(e^{\beta \delta} + 1\right)}\right)}{3 \left(-1 + \frac{1}{\left(1 + e^{- \beta \left(\delta + \theta\right)}\right) \left(e^{\beta \delta} + 1\right)} + \frac{1}{\left(1 + e^{- \beta \delta}\right) \left(e^{\beta \delta} + 1\right)}\right)} \\&\qquad - \frac{8}{3 \left(e^{\beta \delta} + 1\right)^{2} \left(-1 + \frac{1}{\left(1 + e^{- \beta \left(\delta + \theta\right)}\right) \left(e^{\beta \delta} + 1\right)} + \frac{1}{\left(1 + e^{- \beta \delta}\right) \left(e^{\beta \delta} + 1\right)}\right)}.
\end{align*}

\end{comment}

\subsubsection{$N=4$, $t=2$}

For $N=4$, $t=2$, we have
$$
W = \begin{pmatrix}
1&-a&0\\
-c&1&-a\\
0&-d&1\\
\end{pmatrix}
$$

and 

$$
W^{-1} = \begin{pmatrix}
\frac{ad-1}{ac+ad-1}&\frac{-a}{ac+ad-1}&\frac{-a^2}{ac+ad-1}\\
\frac{-c}{ac+ad-1}&\frac{-1}{ac+ad-1}&\frac{-a}{ac+ad-1}\\
\frac{-cd}{ac+ad-1}&\frac{-d}{ac+ad-1}&\frac{ac-1}{ac+ad-1}\\
\end{pmatrix}. 
$$

Thereby, the cost function is
\begin{align*}
E_r(\theta) &= \frac{N^2\theta}{2} \sum_{j=1}^{2}\frac{(W^{-1})_{1j} + (W^{-1})_{N-1,j}}{N-j} 
\\&= 8\theta\Bigg[\frac{(W^{-1}_{11}+W^{-1}_{31})}{3}+\frac{(W^{-1}_{12}+W^{-1}_{32})}{2}\Bigg]
\\&= \frac{8\theta}{3}\Big(\frac{-cd+ad-1}{ac+ad-1} \Big) +  4\theta\Big(\frac{-a-d}{ac+ad-1}\Big)
\\&= \frac{8\theta}{3}\Bigg(\frac{-1+\frac{1}{(1+e^{-x})(e^x+1)}}{-1+\frac{1}{(1+e^{-x})(e^y+1)}+\frac{1}{(1+e^{-x})(e^x+1)}} - \frac{1}{(e^x+1)(e^y+1)\Big(-1+\frac{1}{(1+e^{-x})(e^y+1)}+\frac{1}{(1+e^{-x})(e^x+1)}\Big)} \Bigg) + \\& \qquad 4\theta\Bigg(\frac{1}{(e^y+1)\Big(1-\frac{1}{(1+e^{-x})(e^y+1)}-\frac{1}{(1+e^{-x})(e^x+1)}\Big)} - \frac{1}{(1+e^{-x})\Big(-1+\frac{1}{(1+e^{-x})(e^y+1)}+\frac{1}{(1+e^{-x})(e^x+1)}\Big)} \Bigg),
\end{align*} where $x=\beta(\delta+\theta)$ and $y=\beta\delta$.

\subsubsection{$N=4$, $t=3$}

For $N=4$, $t=3$, we have
$$
W = \begin{pmatrix}
1&-a&0\\
-c&1&-a\\
0&-c&1\\
\end{pmatrix}
$$

and 

$$
W^{-1} = \begin{pmatrix}
\frac{ac-1}{2ac-1}&\frac{-a}{2ac-1}&\frac{-a^2}{2ac-1}\\
\frac{-c}{2ac-1}&\frac{-1}{2ac-1}&\frac{-a}{2ac-1}\\
\frac{-c^2}{2ac-1}&\frac{-c}{2ac-1}&\frac{ac-1}{2ac-1}\\
\end{pmatrix}. 
$$

Thus, the cost function is
\begin{align*}
E_r(\theta) &= \frac{N^2\theta}{2} \sum_{j=1}^{3}\frac{(W^{-1})_{1j} + (W^{-1})_{N-1,j}}{N-j} 
\\&= \frac{N^2\theta}{2}\Bigg[\frac{(W^{-1}_{11}+W^{-1}_{31})}{3}+\frac{(W^{-1}_{12}+W^{-1}_{32})}{2} + (W^{-1}_{13}+W^{-1}_{33})\Bigg]
\\&=  \frac{8\theta}{3}\Big(\frac{c^2+ac-1}{2ac-1} \Big) + 4\theta\Big(\frac{-a+c}{2ac-1}\Big) + 8\theta\Big(\frac{-a^2+ac-1}{2ac-1}\Big)
\\&= \frac{8\theta}{3}\Bigg(\frac{-1+\frac{1}{(1+e^{-x})(e^x+1)}}{-1+\frac{2}{(1+e^{-x})(e^x+1)}} - \frac{1}{\Big(1-\frac{2}{(1+e^{-x})(e^x+1)}\Big)(1+e^x)^2} \Bigg) + \\& \qquad 4\theta\Bigg(\frac{1}{\Big(1-\frac{2}{(1+e^{-x})(e^x+1)}\Big)(e^x+1)} - \frac{1}{\Big(-1+\frac{2}{(1+e^{-x})(e^x+1)}\Big)(e^{-x}+1)}\Bigg) + \\& \qquad 8\theta\Bigg(\frac{-1+\frac{1}{(1+e^{-x})(e^x+1)}}{-1+\frac{2}{(1+e^{-x})(e^x+1)}} - \frac{1}{\Big(-1+\frac{2}{(1+e^{-x})(e^x+1)}\Big)(1+e^{-x})^2} \Bigg),
\end{align*} where $x=\beta(\delta+\theta)$ and $y=\beta\delta$.

\subsection{The $W$ entries for $E_r(\theta)$ with $t = 2$}
\label{sec: Reward t = 2 Ws}

\begin{align*}
(W^{-1}_{1,1})&=\frac{1}{1-(-c)(-a)\frac{y_3}{y_2}}
\\& = \frac{1}{1-ac\frac{y_3}{y_2}},\\
(W^{-1})_{N-1,1}&=(-1)^{N-2}\left(\prod_{k=1}^{N-2} a_{k+1}\right) \frac{y_{N}}{y_{2}} \phi_{1,1}
\\& = (-1)^{N-2} \Bigg(\prod_{k=1}^{N-2} a_{k+1}\Bigg)\frac{1}{y_{2}} \frac{1}{1-ac\frac{y_3}{y_2}}
\\& =(-1)^{N^{-2}}\left(a_{2} \cdot a_{3} \cdot \ldots \cdot a_{N-1}\right) \frac{1}{y_{2}} \frac{1}{1-a_{c} \frac{y_{3}}{y_{2}}} \\
\\& =(-1)^{N-2}(-c) \cdot(-d) \cdot \ldots \cdot(-d) \frac{1}{y_{2}} \frac{1}{1-a c \frac{y_{3}}{y_{2}}} \\
\\& =(-1)^{N-2}(-c)(-d)^{N-3} \frac{1}{y_2} \frac{1}{1-a c \frac{y_{3}}{y_{2}}}
\\& =c d^{N-3} \frac{1}{y_{2}} \frac{1}{1-a c \frac{y_{3}}{y_{2}}},\\
(W^{-1})_{1,2}&=(-1)^{1}\left(\prod_{k=1}^{1} c_{2-k}\right) \frac{z_{0}}{z_{1}} \phi_{2,2} 
\\& =-\left(c_{1}\right) \frac{1}{1} \phi_{2,2} 
\\& =a \phi_{2,2} 
\\& =a \frac{1}{1-(-c)(-a) \frac{z_{0}}{z_{1}}-(-d)(-a) \frac{y_{4}}{y_{3}}} 
\\& =a \frac{1}{1-a c-a d\frac{y_4}{y_3}},\\
(W^{-1})_{N-1,2}&=(-1)^{N-3}\left(\prod_{k=1}^{N-3} a_{2+k}\right) \frac{y_{N}}{y_{3}} \phi_{2,2} 
\\& =(-1)^{N-3}\left(a_{3} \cdot a_{4} \ldots \cdot a_{N-1}\right) \frac{1}{y_{3}} \frac{1}{1-a c-a d \frac{y_{4}}{y_{3}}} 
\\& =(-1)^{N-3}(-d)^{N-3} \frac{1}{y_{3}} \frac{1}{1-a c-a d \frac{y_{4}}{y_{3}}} \\
\\& =d^{N-3}\frac{1}{y_{3}} \frac{1}{1-a c-a d \frac{y_4}{y_3}}.
\end{align*}

\subsection{Behaviour of the cost functions under small mutation with different initial states}
\label{sec: behaviours}

In this subsection, we present the results of our numerical simulation in the case of reward incentives for PGG and in the cases of punishment and hybrid incentives for DG, for the two initial starting states. 

\begin{figure}[H]
\centering
\hspace*{-1cm}  
\includegraphics[width=0.8\textwidth]{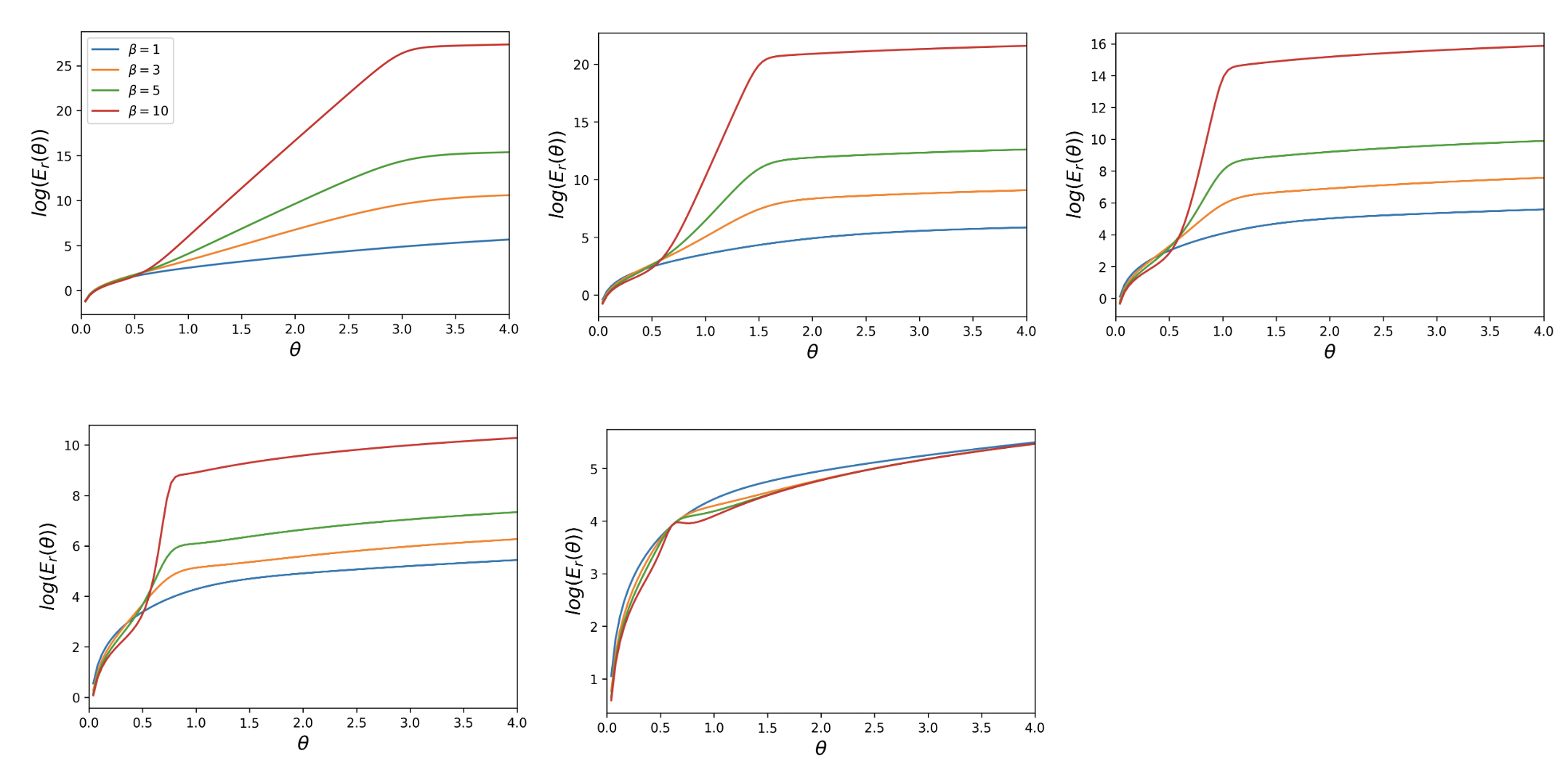}
\caption{Behaviour of the reward cost function $E_{r}(\theta)$, for different thresholds $t$ and strengths of selection $\beta$, for PGG with $N = 6, c = 1, r = 2, n = 3$. The first row corresponds to $t = 1$, $t = 2$, and $t = 3$, while the second one to $t = 4$ and $t = 5$. This is for the assumption that the
population is equally likely to start in the homogeneous state $S_0$ as well as in the homogeneous
state $S_N$.}
\label{fig:varioustPGGold}
\end{figure}

\begin{figure}[H]
\centering
\hspace*{-1cm}  
\includegraphics[width=0.8\textwidth]{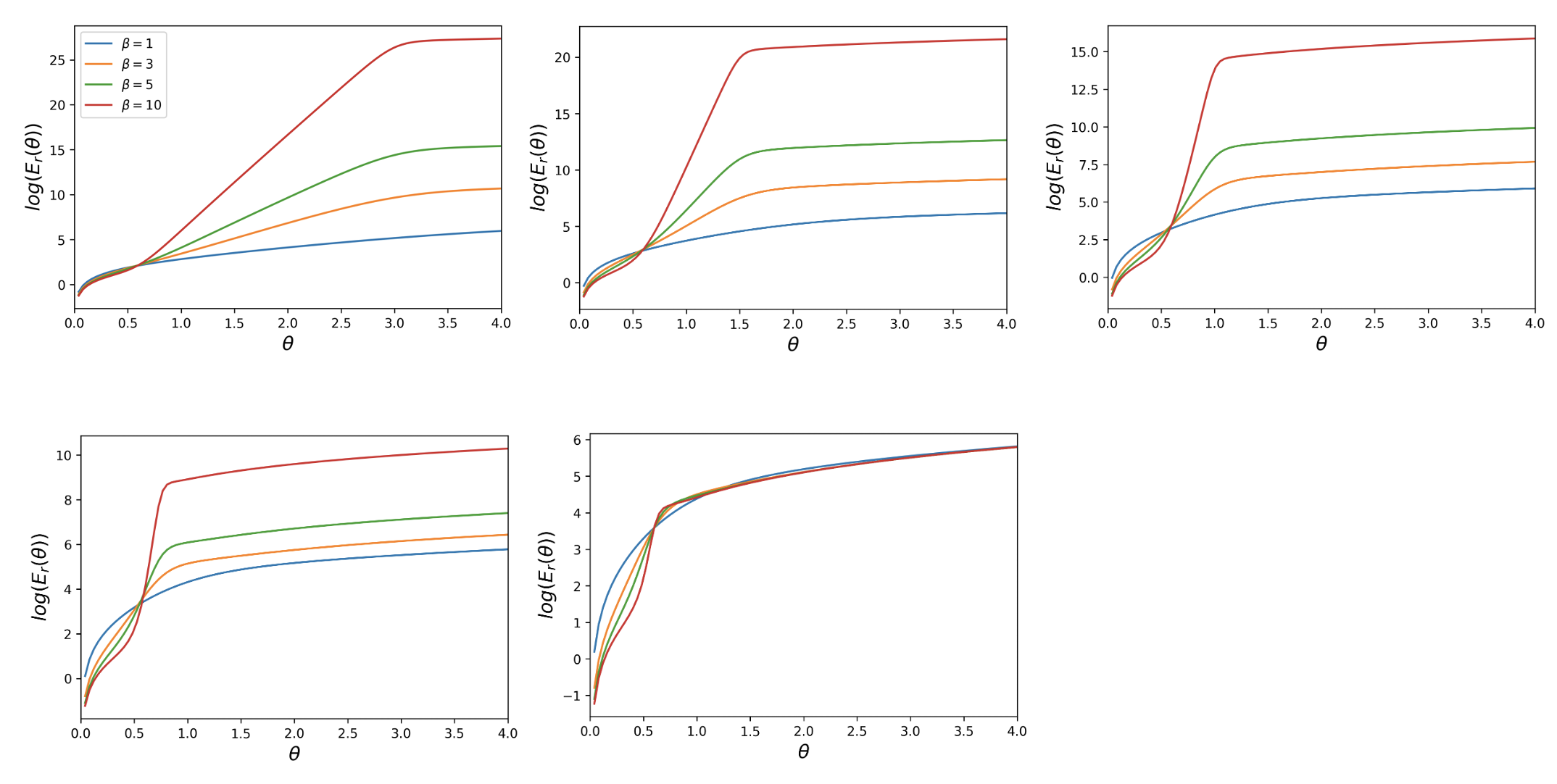}
\caption{Behaviour of the reward cost function $E_{r}(\theta)$, for different thresholds $t$ and strengths of selection $\beta$, for PGG with $N = 6, c = 1, r = 2, n = 3$. The first row corresponds to $t = 1$, $t = 2$, and $t = 3$, while the second one to $t = 4$ and $t = 5$. This is for the assumption that the population is expected to start with all defectors, i.e. in state $S_0$.}
\label{fig:varioustPGGnew}
\end{figure}

\begin{figure}[H]
\centering
\hspace*{-1cm}  
\includegraphics[width=0.8\textwidth]{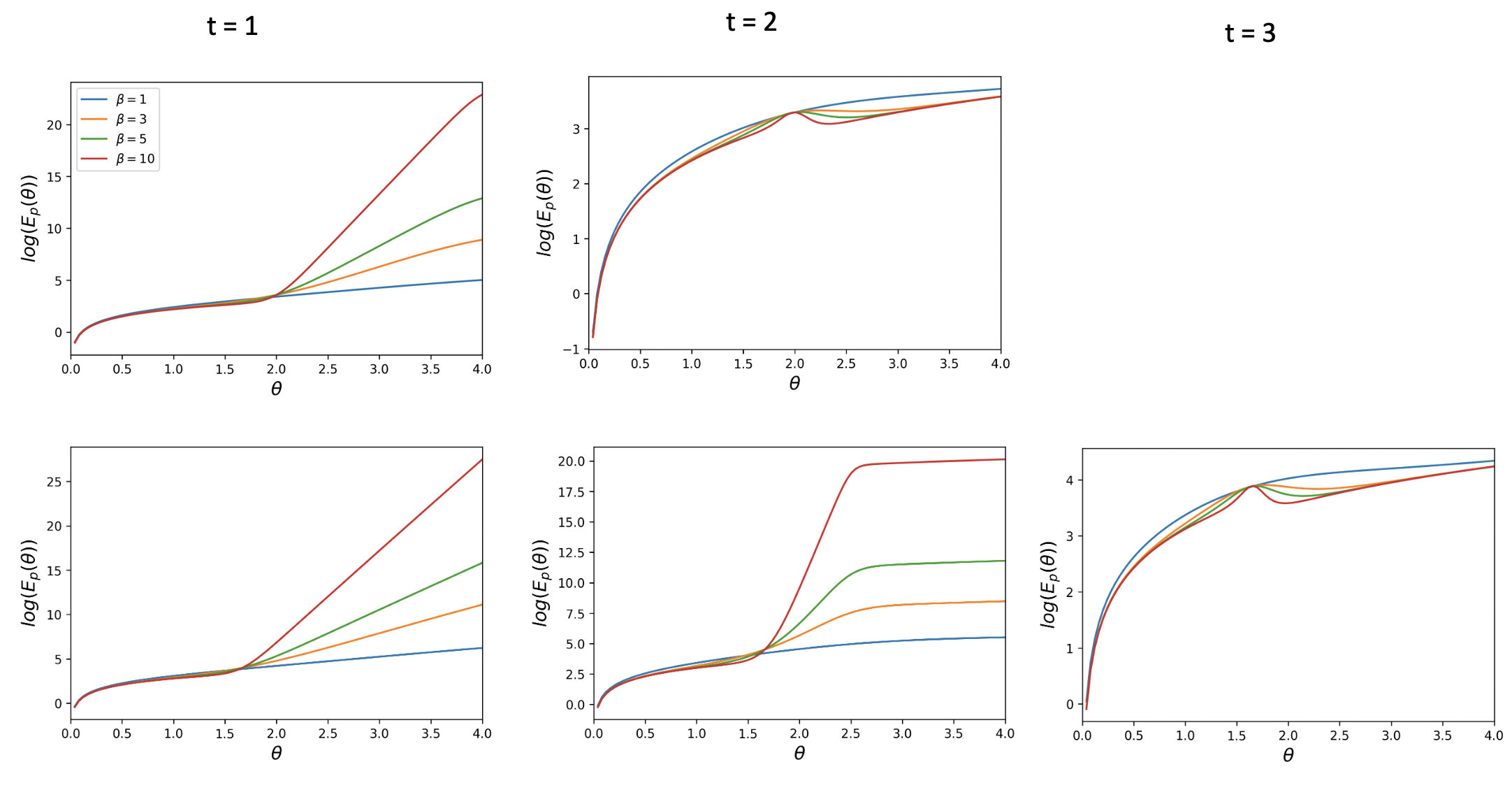}
\caption{Behaviour of the punishment cost function $E_{p}(\theta)$, for different     thresholds $t$ and strengths of selection $\beta$, for DG with $B = 2, c = 1$. The first row corresponds to $N = 3$ and the second one to $N = 4$. The leftmost column corresponds to $t = 1$, the middle one to $t = 2 $, the rightmost one to $t = 3$. The behaviour of the cost function for PGG is similar. This is for the assumption that the
population is equally likely to start in the homogeneous state $S_0$ as well as in the homogeneous
state $S_N$.}
\label{fig:varioustDGpunishmentold}
\end{figure}

\begin{figure}[H]
\centering
\hspace*{-1cm}  
\includegraphics[width=0.8\textwidth]{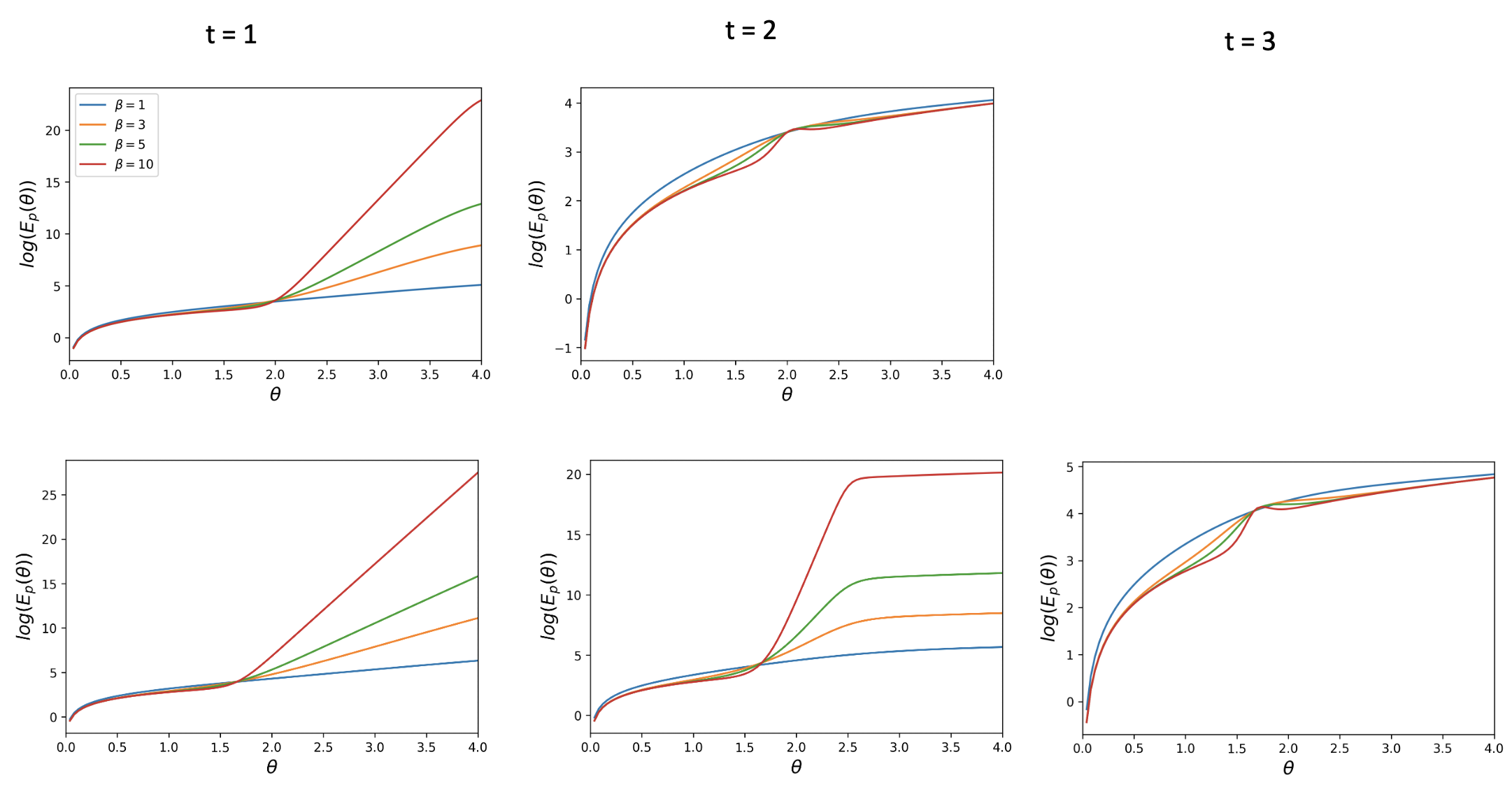}
\caption{Behaviour of the punishment cost function $E_{p}(\theta)$, for different     thresholds $t$ and strengths of selection $\beta$, for DG with $B = 2, c = 1$. The first row corresponds to $N = 3$ and the second one to $N = 4$. The leftmost column corresponds to $t = 1$, the middle one to $t = 2 $, the rightmost one to $t = 3$. The behaviour of the cost function for PGG is similar. This is for the assumption that the population is expected to start with all defectors, i.e. in state $S_0$.}
\label{fig:varioustDGpunishmentnew}
\end{figure}

\begin{figure}[H]
\centering
\hspace*{-1cm}  
\includegraphics[width=0.8\textwidth]{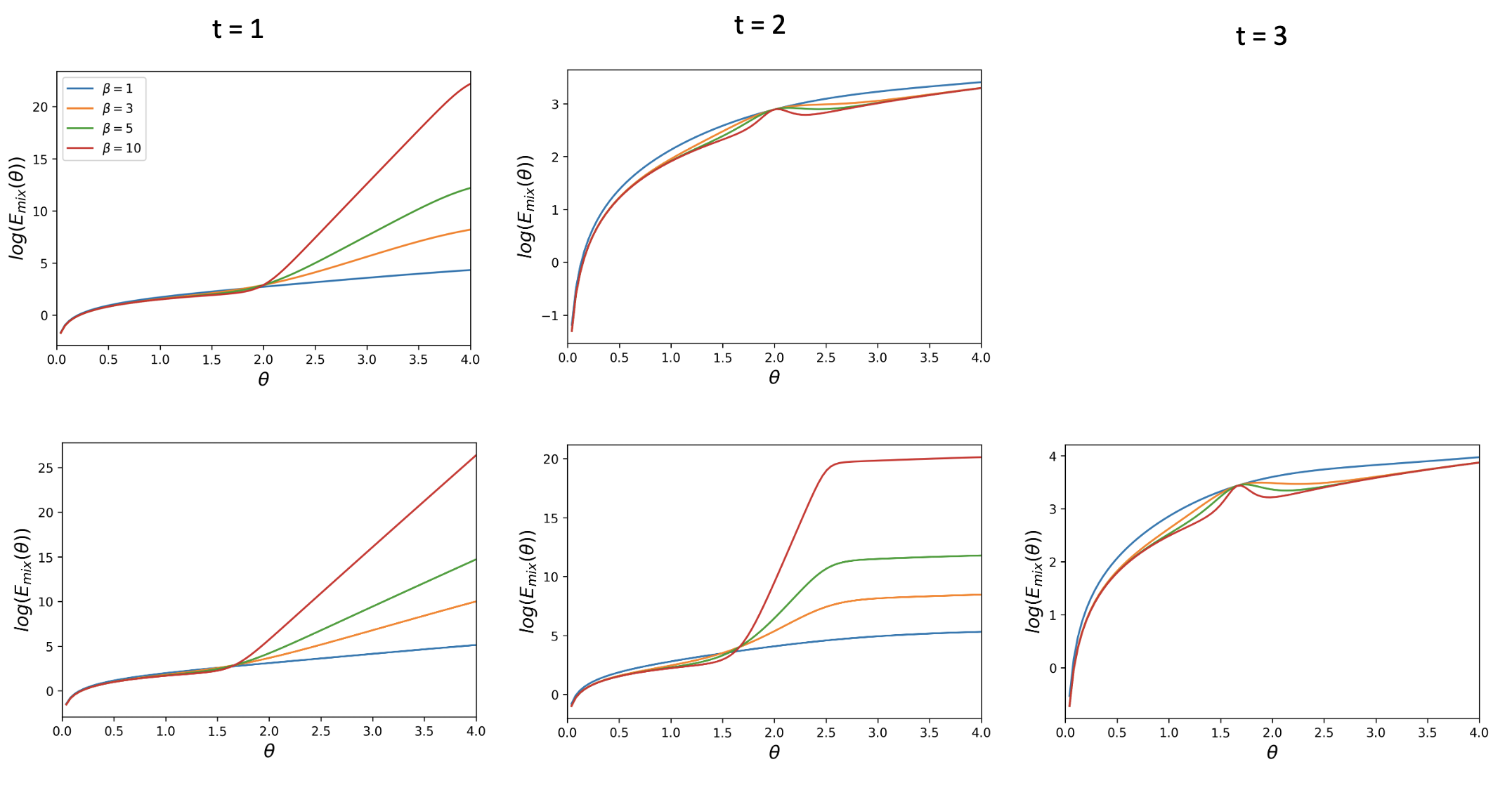}
\caption{Behaviour of the hybrid incentive cost function (mix of reward and punishment) $E_{mix}(\theta)$ (where $a = 1, b = 1$), for different thresholds $t$ and strengths of selection $\beta$, for DG with $B = 2, c = 1$. The first row corresponds to $N = 3$ and the second one to $N = 4$. The leftmost column corresponds to $t = 1$, the middle one to $t = 2 $, the rightmost one to $t = 3$. The behaviour of the cost function for PGG is similar. This is for the assumption that the
population is equally likely to start in the homogeneous state $S_0$ as well as in the homogeneous
state $S_N$.}
\label{fig:varioustDGrewardandpunishmentold}
\end{figure}

\begin{figure}[H]
\centering
\hspace*{-1cm}  
\includegraphics[width=0.8\textwidth]{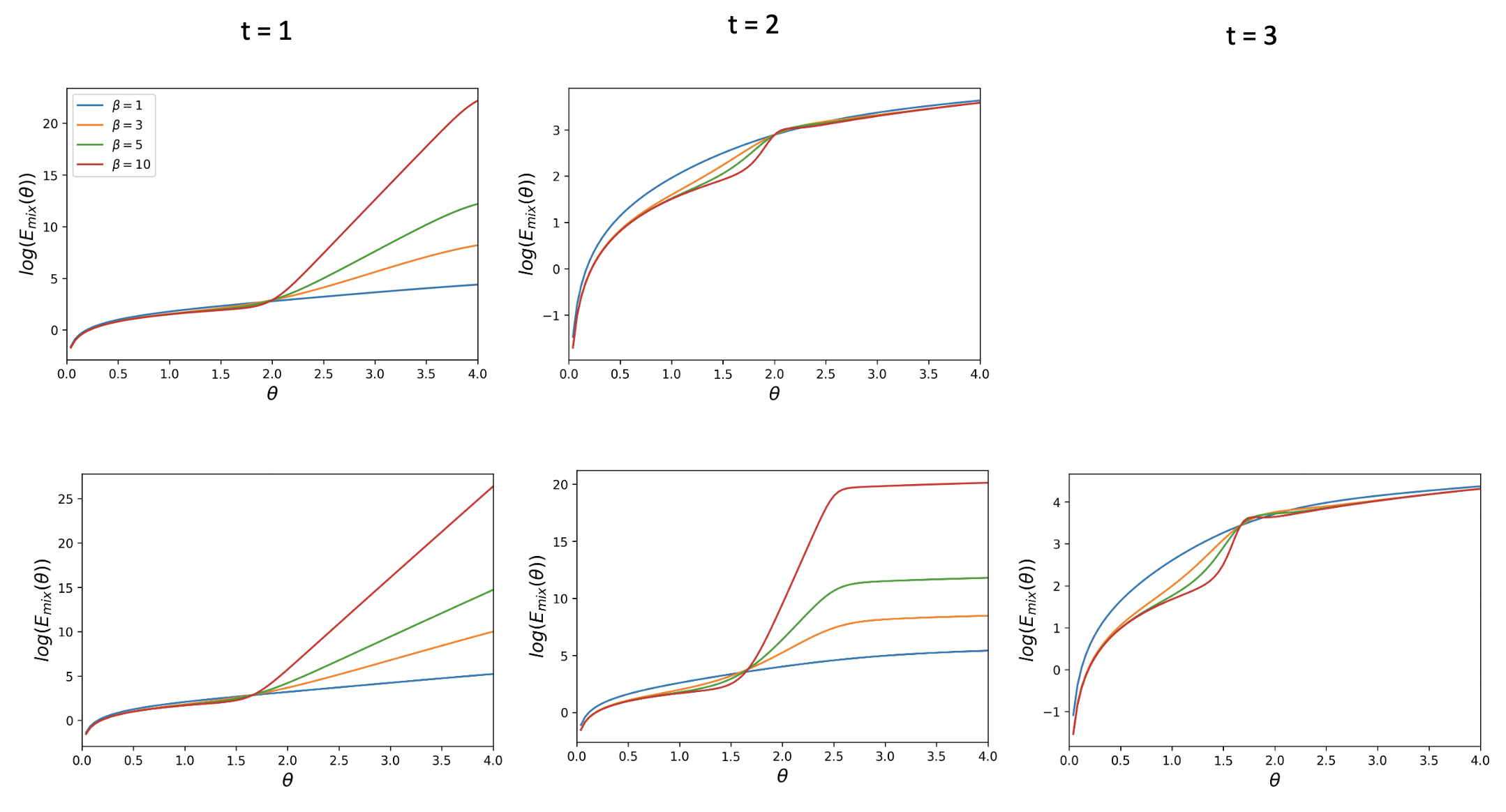}
\caption{Behaviour of the hybrid incentive cost function (mix of reward and punishment) $E_{mix}(\theta)$ (where $a = 1, b = 1$), for different thresholds $t$ and strengths of selection $\beta$, for DG with $B = 2, c = 1$. The first row corresponds to $N = 3$ and the second one to $N = 4$. The leftmost column corresponds to $t = 1$, the middle one to $t = 2 $, the rightmost one to $t = 3$. The behaviour of the cost function for PGG is similar. This is for the assumption that the population is expected to start with all defectors, i.e. in state $S_0$.}
\label{fig:varioustDGrewardandpunishmentnew}
\end{figure}

\section*{Acknowledgements} The research of M.H.D was supported by EPSRC Grants EP/V038516/1, EP/Y008561/1, and a Royal International Exchange Grant IES-R3-223047. C.M.D is supported by an EPSRC Studentship. T.A.H. is supported by EPSRC (grant EP/Y00857X/1) and the Future of Life Institute. 
\section*{Data Availability Statement}
Data sharing is not applicable to this article as no datasets were generated or analysed during the current study.

\bibliographystyle{plain}
\bibliography{references}

\begin{thebibliography}{10}

\bibitem{atkins2019prosocial}
Paul~WB Atkins, David~Sloan Wilson, and Steven~C Hayes.
\newblock {\em Prosocial: Using evolutionary science to build productive,
  equitable, and collaborative groups}.
\newblock New Harbinger Publications, 2019.

\bibitem{moralpref}
Valerio Capraro and Matja{\v{z}} Perc.
\newblock Mathematical foundations of moral preferences.
\newblock {\em Journal of the Royal Society interface}, 18(175):20200880, 2021.

\bibitem{chen2014optimal}
Xiaojie Chen and Matja{\v{z}} Perc.
\newblock {Optimal distribution of incentives for public cooperation in
  heterogeneous interaction environments}.
\newblock {\em Frontiers in behavioral neuroscience}, 8:248, 2014.

\bibitem{cimpeanu2021cost}
Theodor Cimpeanu, Cedric Perret, and The~Anh Han.
\newblock {Cost-efficient interventions for promoting fairness in the ultimatum
  game}.
\newblock {\em Knowledge-Based Systems}, 233:107545, 2021.

\bibitem{cimpeanu2023does}
Theodor Cimpeanu, Francisco~C Santos, and The~Anh Han.
\newblock Does spending more always ensure higher cooperation? an analysis of
  institutional incentives on heterogeneous networks.
\newblock {\em Dynamic Games and Applications}, pages 1--20, 2023.

\bibitem{DuongDurbacHan2022}
M.~H. Duong, C.~M. Durbac, and T.~A. Han.
\newblock Cost optimisation of hybrid institutional incentives for promoting
  cooperation in finite populations.
\newblock {\em J. Math. Biol.}, 87(77), 2023.

\bibitem{duong2021cost}
Manh~Hong Duong and The~Anh Han.
\newblock Cost efficiency of institutional incentives for promoting cooperation
  in finite populations.
\newblock {\em Proceedings of the Royal Society A}, 477(2254):20210568, 2021.

\bibitem{flores2024evolution}
Lucas~S Flores and The~Anh Han.
\newblock Evolution of commitment in the spatial public goods game through
  institutional incentives.
\newblock {\em Applied Mathematics and Computation}, 473:128646, 2024.

\bibitem{gois2019reward}
Ant{\'o}nio~R. G{\'o}is, Fernando~P. Santos, Jorge~M. Pacheco, and Francisco~C.
  Santos.
\newblock Reward and punishment in climate change dilemmas.
\newblock {\em Sci. Rep.}, 9(1):1--9, 2019.

\bibitem{gurerk}
Özgür Gürerk, Bernd Irlenbusch, and Bettina Rockenbach.
\newblock The competitive advantage of sanctioning institutions.
\newblock {\em Science (New York, N.Y.)}, 312:108--11, 05 2006.

\bibitem{han2018fostering}
The~Anh Han, Simon Lynch, Long Tran-Thanh, and Francisco~C Santos.
\newblock Fostering cooperation in structured populations through local and
  global interference strategies.
\newblock In {\em Proceedings of the 27th International Joint Conference on
  Artificial Intelligence}, pages 289--295, 2018.

\bibitem{han2018cost}
The~Anh Han and Long Tran-Thanh.
\newblock Cost-effective external interference for promoting the evolution of
  cooperation.
\newblock {\em Scientific reports}, 8(1):1--9, 2018.

\bibitem{hauert2007}
Christoph Hauert, Arne Traulsen, Hannelore Brandt, Martin~A. Nowak, and Karl
  Sigmund.
\newblock Via freedom to coercion: The emergence of costly punishment.
\newblock {\em Science}, 316(5833):1905--1907, 2007.

\bibitem{hu2020rewarding}
Liwen Hu, Nanrong He, Qifeng Weng, Xiaojie Chen, and Matja{\v{z}} Perc.
\newblock Rewarding endowments lead to a win-win in the evolution of public
  cooperation and the accumulation of common resources.
\newblock {\em Chaos, Solitons \& Fractals}, 134:109694, 2020.

\bibitem{hua2024coevolutionary}
Shijia Hua and Linjie Liu.
\newblock Coevolutionary dynamics of population and institutional rewards in
  public goods games.
\newblock {\em Expert Systems With Applications}, 237:121579, 2024.

\bibitem{huang1997}
Y~Huang and W~F McColl.
\newblock Analytical inversion of general tridiagonal matrices.
\newblock {\em Journal of Physics A: Mathematical and General}, 30(22):7919,
  1997.

\bibitem{kemeny1976finite}
Jim Kemeny.
\newblock Perspectives on the micro-macro distinction.
\newblock {\em The Sociological Review}, 24(4):731--752, 1976.

\bibitem{liu2022effects}
Linjie Liu and Xiaojie Chen.
\newblock Effects of interconnections among corruption, institutional
  punishment, and economic factors on the evolution of cooperation.
\newblock {\em Applied Mathematics and Computation}, 425:127069, 2022.

\bibitem{nowak}
Martin~A. Novak.
\newblock {\em Evolutionary Dynamics: Exploring the Equations of Life}.
\newblock Harvard University Press, 2006.

\bibitem{nowak2006}
Martin~A. Nowak.
\newblock Five rules for the evolution of cooperation.
\newblock {\em Science}, 314(5805):1560--1563, 2006.

\bibitem{nowak2004emergence}
Martin~A. Nowak, Akira Sasaki, Christine Taylor, and Drew Fudenberg.
\newblock Emergence of cooperation and evolutionary stability in finite
  populations.
\newblock {\em Nature}, 428(6983):646--650, 2004.

\bibitem{ostrom2009understanding}
Elinor Ostrom.
\newblock {\em Understanding institutional diversity}.
\newblock Princeton university press, 2005.

\bibitem{perc2017statistical}
Matja{\v{z}} Perc, Jillian~J. Jordan, David~G. Rand, Zhen Wang, Stefano
  Boccaletti, and Attila Szolnoki.
\newblock {Statistical physics of human cooperation}.
\newblock {\em Physics Reports}, 687:1--51, 2017.

\bibitem{rand2013human}
David~G. Rand and Martin~A. Nowak.
\newblock Human cooperation.
\newblock {\em Trends in cognitive sciences}, 17(8):413--425, 2013.

\bibitem{rand2013evolution}
David~G Rand, Corina~E. Tarnita, Hisashi Ohtsuki, and Martin~A. Nowak.
\newblock Evolution of fairness in the one-shot anonymous ultimatum game.
\newblock {\em Proceedings of the National Academy of Sciences},
  110(7):2581--2586, 2013.

\bibitem{rockenbach}
Bettina Rockenbach and Manfred Milinski.
\newblock The efficient interaction of indirect reciprocity and costly
  punishment.
\newblock {\em Nature}, 444:718--23, 01 2007.

\bibitem{sasaki2012take}
Tatsuya Sasaki, {\AA}ke Br{\"a}nnstr{\"o}m, Ulf Dieckmann, and Karl Sigmund.
\newblock The take-it-or-leave-it option allows small penalties to overcome
  social dilemmas.
\newblock {\em Proceedings of the National Academy of Sciences},
  109(4):1165--1169, 2012.

\bibitem{carrotstick}
Tatsuya Sasaki, Xiaojie Chen, {\AA}ke Br{\"a}nnstr{\"o}m, and Ulf Dieckmann.
\newblock First carrot, then stick: How the adaptive hybridization of
  incentives promotes cooperation.
\newblock {\em Journal of the Royal Society Interface}, 12:20140935, 01 2015.

\bibitem{sigmund2010calculus}
Karl Sigmund.
\newblock The calculus of selfishness.
\newblock In {\em The Calculus of Selfishness}. Princeton University Press,
  2010.

\bibitem{sigmundinstitutions}
Karl Sigmund, Hannelore De~Silva, Arne Traulsen, and Christoph Hauert.
\newblock Social learning promotes institutions for governing the commons.
\newblock {\em Nature}, 466:7308, 2010.

\bibitem{sun2021combination}
Weiwei Sun, Linjie Liu, Xiaojie Chen, Attila Szolnoki, and V{\'\i}tor~V
  Vasconcelos.
\newblock Combination of institutional incentives for cooperative governance of
  risky commons.
\newblock {\em Iscience}, 24(8), 2021.

\bibitem{2ndfreeriding}
Attila Szolnoki and Matja\ifmmode \check{z}\else~\v{z}\fi{} Perc.
\newblock Second-order free-riding on antisocial punishment restores the
  effectiveness of prosocial punishment.
\newblock {\em Phys. Rev. X}, 7:041027, Oct 2017.

\bibitem{traulsen2006}
Arne Traulsen and Martin~A. Nowak.
\newblock Evolution of cooperation by multilevel selection.
\newblock {\em Proceedings of the National Academy of Sciences},
  103(29):10952--10955, 2006.

\bibitem{van2014reward}
Paul~A.M. Van~Lange, Bettina Rockenbach, and Toshio Yamagishi.
\newblock {\em Reward and punishment in social dilemmas}.
\newblock Oxford University Press, 2014.

\bibitem{wang2024partial}
Jianwei Wang, Wenhui Dai, Yanfeng Zheng, Fengyuan Yu, Wei Chen, and Wenshu Xu.
\newblock Partial intervention promotes cooperation and social welfare in
  regional public goods game.
\newblock {\em Chaos, Solitons \& Fractals}, 184:114991, 2024.

\bibitem{wang2023optimally}
Shengxian Wang, Ming Cao, and Xiaojie Chen.
\newblock Optimally combined incentive for cooperation among interacting agents
  in population games, 2023.

\bibitem{wang2019exploring}
Shengxian Wang, Xiaojie Chen, and Attila Szolnoki.
\newblock Exploring optimal institutional incentives for public cooperation.
\newblock {\em Communications in Nonlinear Science and Numerical Simulation},
  79:104914, 2019.

\bibitem{wangdecentralised}
Shengxian Wang, Xiaojie Chen, Zhilong Xiao, and Attila Szolnoki.
\newblock Decentralized incentives for general well-being in networked public
  goods game.
\newblock {\em Applied Mathematics and Computation}, 431:127308, 2022.

\bibitem{wang2023optimization}
Shengxian Wang, Xiaojie Chen, Zhilong Xiao, Attila Szolnoki, and V{\'\i}tor~V
  Vasconcelos.
\newblock Optimization of institutional incentives for cooperation in
  structured populations.
\newblock {\em Journal of the Royal Society Interface}, 20(199):20220653, 2023.

\bibitem{wang2021}
Shengxian Wang, Linjie Liu, and Xiaojie Chen.
\newblock Incentive strategies for the evolution of cooperation: Analysis and
  optimization.
\newblock {\em Europhysics Letters}, 136(6):68002, 2021.

\bibitem{xia2023reputation}
Chengyi Xia, Juan Wang, Matja{\v{z}} Perc, and Zhen Wang.
\newblock Reputation and reciprocity.
\newblock {\em Physics of Life Reviews}, 2023.

\bibitem{zisis2015generosity}
Ioannis Zisis, Sibilla Di~Guida, The~Anh Han, Georg Kirchsteiger, and Tom
  Lenaerts.
\newblock Generosity motivated by acceptance-evolutionary analysis of an
  anticipation game.
\newblock {\em Scientific reports}, 5(1):1--11, 2015.

\end{thebibliography}
\end{document}